\documentclass[11pt]{article}

\usepackage{arxiv}

\usepackage{mathpazo}
\usepackage{euscript}
\usepackage{calrsfs}
\usepackage{bbold}
\usepackage{microtype}

\usepackage[utf8]{inputenc} 
\usepackage[T1]{fontenc}    
\usepackage{hyperref}       
\usepackage{url}            
\usepackage{booktabs}       
\usepackage{amsfonts}       

\usepackage{cite}
\usepackage{tabularx}
\usepackage[english]{babel}
\usepackage{amsmath}
\usepackage{amsthm}
\usepackage{amssymb}
\usepackage{caption}
\usepackage{subcaption}
\usepackage{graphicx}
\graphicspath{{./figs/}}
\usepackage{xcolor}
\usepackage{tocloft}
\usepackage{thmtools}
\usepackage{algorithm2e}

\usepackage{listings}

\theoremstyle{plain}
\newtheorem{theorem}{Theorem}[section]

\newtheorem{invariant}[theorem]{Invariant}

\newtheorem{corollary}[theorem]{Corollary}
\newtheorem{lemma}[theorem]{Lemma}

\newcommand{\R}{\mathbb{R}}
\newcommand{\reals}{\mathbb{R}}
\newcommand{\eps}{\varepsilon}
\newcommand{\A}{\mathcal{A}}
\newcommand{\F}{\mathcal{F}}
\newcommand{\G}{\mathcal{G}}
\newcommand{\U}{\mathcal{U}}
\newcommand{\C}{\mathcal{C}}
\newcommand{\cL}{\mathcal{L}}
\newcommand{\cU}{\mathcal{U}}
\def\lclenv{\mathsf{L}}
\def\edges{\mathcal{E}}
\def\bd{\partial}
\def\area{\mathop{\mathrm{Area}}}
\def\GG{\mathbb{G}}
\def\ucircle{\mathsf{s}^+}

\title{Maintaining the Union of Unit Discs under 
Insertions with Near-Optimal Overhead\thanks{A preliminary
version appeared as Pankaj K.~Agarwal,
Ravid Cohen, Dan Halperin, and
Wolfgang Mulzer. \emph{Maintaining the Union of Unit Discs under 
Insertions with Near-Optimal Overhead}.
Proceedings of the 35th International Symposium on Computational 
Geometry (SoCG), pp.~26:1--26:15, 2019.
Work by P.A. has been supported by
NSF under grants CCF-15-13816, CCF-15-46392, and IIS-14-08846, by ARO 
grant W911NF-15-1-0408, and by grant 2012/229 from the U.S.-Israel 
Binational Science Foundation. Work by D.H.\ and R.C.\ has been supported 
in part by the Israel Science
Foundation (grant no.~1736/19), 
by NSF/US-Israel-BSF (grant no.~2019754),
by the Israel Ministry of Science and Technology (grant no.~103129), 
by the Blavatnik Computer Science Research Fund, and by
the Yandex Machine Learning Initiative for Machine Learning 
at Tel Aviv University.
Work by 
W.M.~has been partially supported by ERC STG 757609 and GIF grant 1367/2016.

}
}

\author{
  Pankaj K. Agarwal \\
  Department of Computer Science,\\ 
  Duke University, Durham, NC 27708, USA\\ 
  \And
  Ravid Cohen \\
  School of Computer Science,\\
  Tel-Aviv University, Tel-Aviv 69978, Israel
  \And
  Dan Halperin \\
  School of Computer Science,\\
  Tel-Aviv University, Tel-Aviv 69978, Israel
  \And
  Wolfgang Mulzer \\
  Institut f\"ur Informatik,\\ 
  Freie Universit\"at Berlin, D-14195 Berlin, Germany\\
}
\begin{document}
\maketitle

\begin{abstract}
We present efficient dynamic data structures for 
maintaining the union of unit discs and the lower envelope of pseudo-lines in the plane.
More precisely, we present three main results in this paper:
\begin{itemize}
	\item[(i)]  We present a linear-size data structure to maintain the union of a set 
of unit discs under insertions. It can insert a disc and update the union in $O((k+1) \log^2 n)$ time, where $n$ is the
current number of unit discs and $k$ is the combinatorial complexity of the structural
change in the union due to the insertion of the new disc.  It can also compute, within the same time bound, the area of the union after the insertion of each disc.
	\item[(ii)]  We propose a linear-size data structure for maintaining the lower envelope of a set of $x$-monotone pseudo-lines.
		It can handle insertion/deletion of a pseudo-line in $O(\log^2 n)$  time;
		for a query point $x_0\in\reals$, it can report, in $O(\log n)$ time, the point on the lower envelope with $x$-coordinate $x_0$; and 
		for a query point $q\in\R^2$, it can return all $k$ pseudo-lines lying below $q$ in time $O(\log n+k\log^2 n)$.
	\item[(iii)] We present a linear-size data structure for storing a set of circular arcs of unit radius
		(not necessarily on the boundary of the union of the corresponding discs), so that for a query unit disc $D$,
		all input arcs intersecting $D$ can be reported in $O(n^{1/2+\eps} + k)$ time, where $k$ is the output size and 
		$\eps > 0$ is an arbitrarily small constant. A unit-circle arc can be inserted or deleted in $O(\log^2 n)$ time.
\end{itemize}
\end{abstract}

\section{Introduction}

\paragraph{Problem statement.}
Let $S=\{p_1,\ldots,p_n\}$ be set of $n$ points in $\R^2$, let $D(p_i)$ be the unit disc centered at $p_i$, 
and let $U := U(S) := \bigcup_{p\in S} D(p)$ be the union of the unit discs centered at the points of $S$. We wish
to maintain the boundary $\partial U$ of $U$, 
as new points are added to $S$. In particular, we wish to maintain (i) the set of edges on $\bd U$ and (ii) the area of $U$.

The efficient manipulation of collections of unit discs in the plane is a widely and frequently studied
topic, e.g., in the context of sensor networks, where every disc represents the area covered by a sensor.
In our setting, we are motivated by the problem of multiple agents traversing a region in search of a 
particular target~\cite{abs-1904-02895}. We are interested in
investigating the pace of coverage as the agents move, and we wish to estimate at each stage the overall 
area that has been covered so far.  The simulation is discretized, i.e., each agent is modeled by a unit disc whose 
motion is simulated by changing its location at fixed time steps. In other words, we are receiving a stream $\{p_1, p_2, ..., p_i\}$ of 
points in $\R^2$. When the next point $p_{i+1}$ arrives, we want to quickly compute the area of $D(p_{i+1}) \setminus \bigcup_{j\le i} D(p_j)$.

\begin{figure}[ht]
	\centering
    \includegraphics{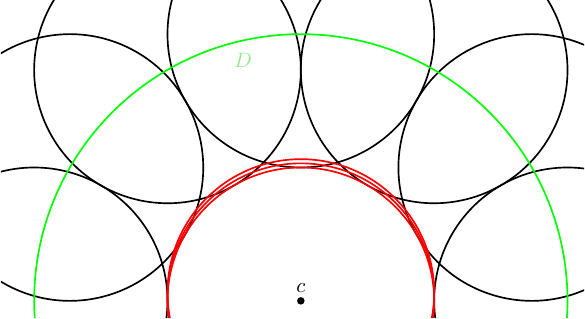}
    \caption{An instance in which the union boundary of a set of
    unit discs in the plane undergoes $\Omega(n^2)$ 
    combinatorial changes during $n$ insertions. The black discs are inserted first, and then red discs are inserted from bottom 
	to top.}
	\label{f:worst_case_example}
\end{figure}

It is known that even for discs of arbitrary radii, the 
boundary $\partial U$ has $O(n)$ vertices and edges~\cite{KedemLPS86}, 
and that $\partial U$ can be computed in $O(n\log n)$ time 
using \emph{power diagrams}~\cite{Aurenhammer1988}. An incremental
algorithm~\cite{spirakis1983very} can maintain $\partial U$
under $n$ insertions in total
time $O(n^2)$.
This is worst-case optimal, as the total amount of structural change to $\partial U$ 
under a sequence of $n$ insertions can be $\Omega(n^2)$ in the worst case. For instance, refer to Figure~\ref{f:worst_case_example}.
    Let $D$ be a disc of radius $2$ centered at the origin (green). We first insert $n/2$ unit discs with equidistant centers on $\partial D$ (black).
    Next, we insert $n/2$ (red) unit discs such that the center of the $i$th red disc is $(0, \eps i)$ on the $y$-axis, for some sufficiently small constant  $\eps > 0$. The insertion of each of the last $n/2$ discs creates $n$ vertices on the 
    union of the discs inserted so far. Our goal is thus to develop an 
 \emph{output-sensitive} algorithm that uses $O(n)$ space and updates $U$ in time proportional to 
the number of changes in vertices and edges of $\partial U$ due to the insertion of a new disc.

Maintaining the edges of $\bd U$ requires answering intersection-searching queries of the
following form: Given a collection $\C$ of unit-radius circular arcs
that comprise $\partial U$ and a query unit disc $D$, report the arcs in $\C$ that intersect $D$. Developing an efficient data structure for this intersection-searching problem led us to study the following problem, which is interesting in its own right:
A set of \emph{pseudo-lines} is a set of
bi-infinite simple curves in the plane such that each pair of  curves
intersect in exactly one point and they cross at that point. 
Given a set $E$ of $x$-monotone pseudo-lines in the plane, we wish to maintain their
lower envelope $\cL(E)$ (see Section~\ref{sec:dynamic_lower_env} below for the definition) under insertions and 
deletions of pseudo-lines, such that for a point $x_0\in\reals$, the point on $\cL(E)$ with $x$-coordinate $x_0$ can be reported 
quickly.

\paragraph{Related work.}
Arrangements of pseudo-lines have been studied extensively in discrete and computational geometry;
see, e.g., the classic monograph by Gr\"unbaum~\cite{Gru} and recent surveys~\cite{fg-pa-18,hs-a-18}
for a review of combinatorial bounds and algorithms involving arrangements of pseudo-lines.
For the case of lines (rather than pseudo-lines), the celebrated result by Overmars and van Leeuwen~\cite{Overmars1981} can maintain the lower envelope in $O(\log^2 n)$ time under insertion and deletion of lines.
This bound has been improved over the last two 
decades~\cite{BrodalJ02,Chan01,HS,BrodalJ00,KaplanTT01}; these
improvements are, however, not directly applicable for pseudo-lines. If only insertions are performed, then the data structure by Preparata~\cite{Pr79} can be extended to maintain the lower envelope of a set of pseudo-lines in $O(\log n)$ time per update.

A series of papers have developed powerful general data structures for maintaining the lower envelopes of a set of curves of bounded description complexity, based on shallow
cuttings~\cite{AgarwalM95,Chan10,KaplanMRSS20,Liu20}.
Many of these data structures also work in $\R^3$.
However, the power of these data structures comes at a significant cost:
the algorithms are quite involved, the performance
guarantees are in the expected and amortized sense, and the
operations have (comparatively) large polylogarithmic
running times. For pseudo-lines, Chan's
method~\cite{Chan10}, with
improvements by Kaplan et
al.~\cite{KaplanMRSS20}, yields
$O(\log^3 n)$ amortized expected insertion time, 
$O(\log^5 n)$ amortized expected deletion time, and 
$O(\log^2 n)$ worst-case query time.  An interesting open question has been whether the Overmars-van-Leeuwen data structure can be extended to maintaining the lower envelope of a set of pseudo-lines.

As mentioned above, a variety of applications have motivated the study of arrangements of unit discs. It is known that the algorithms by  Overmars-van-Leeuwen~\cite{Overmars1981} and by Preparata~\cite{Pr79} for maintaining the intersection of halfplanes can be extended to maintaining the intersection of unit discs within the same time bound. In contrast, maintaining the union of a set of unit discs is more involved and much less is known about this problem. 
The partial rebuilding technique by Bentley and Saxe~\cite{BS80} leads to a linear-size semidynamic data structure for 
maintaining the union of 
unit discs under insertions that can determine in $O(\log^2 n)$ time whether a query point lies in their union; 
a unit disc can be inserted in $O(\log^2 n)$ time. 
Recently de Berg~et~al.~\cite{BBJW21} improved the update and query time to $O(\log n)$.
However, neither of these two approaches can be adapted to maintain the boundary of the union of unit 
discs in output-sensitive manner or to maintain the area of the union under insertion of unit discs. 

Chan~\cite{Chan20a} presented a data structure that can maintain the volume of the convex hull of a set of points in $\reals^3$ in sublinear time. Notwithstanding a close relationship between the union of discs in $\reals^2$ and the convex hull of a point set in $\reals^3$,
it is not clear how to extend his data structure for maintaining the area of the union of unit discs in sublinear time, even if 
we only perform insertions.

We conclude this discussion by noting that there has been extensive work on a variety of intersection-searching problems, in which
we wish to preprocess a set of geometric objects into a data structure so that all objects intersected by a 
query object can be reported efficiently. These data structures typically reduce the problem to simplex or semialgebraic range searching and are based on multi-level partition trees; see, e.g., the recent survey by Agarwal~\cite{a-rs-18} for a review; see 
also~\cite{AvKO,APS,GJS}.

\paragraph{Our results.} This paper contains the following three main results:

\textbf{\textit{Lower envelope of pseudo-lines.}} Our first result is a fully dynamic linear-size data structure for maintaining the lower envelope of a set of $x$-monotone pseudo-lines with $O(\log^2 n)$ 
update time and $O(\log n)$ query time.
Additionally, it can also report all $k$ pseudo-lines
lying below a query point in $O(\log n + k\log^2 n)$ time.
An adaptation of the Overmars-van-Leeuwen data structure~\cite{Overmars1981}, it is more efficient and considerably simpler than the 
existing dynamic data structures for maintaining lower envelopes of pseudo-lines.
The key innovation is a new procedure for finding the
intersection between two lower envelopes of planar pseudo-lines
in $O(\log n)$ time, using \emph{tentative} binary search, where 
each pseudo-line in one envelope is ``smaller'' than every 
pseudo-line in the other envelope, 
in a sense to be made precise below.

\textbf{\textit{Union of unit discs.}} 
Our second result, which is the main result of the paper,
is a linear-size data structure for updating $\bd U$, the boundary of the union of unit discs, in 
$O((k+1)\log^2 n)$ time, per insertion of a disc,  where $k$ is the combinatorial complexity of the structural change to $\partial U$ due to the insertion (see Section~\ref{sec:union_maintain}). 
We use this data structure to compute the change in the area of the union in additional $O((k+1)\log n)$ time, after having computed the changes in $\bd U$.
At the heart of our data structure is a semi-dynamic data structure for reporting all $k$ edges of $\bd U$ that intersect a query unit disc in $O(\log n+k\log^2 n)$ time. Roughly speaking, we draw a uniform grid of diameter $1$. For each grid cell $C$, we 
clip the edges of $\bd U$ within $C$. Let $E_C$ be the set of (clipped) edges of $\bd U$ lying inside $C$. 
For an edge $e\in E_C$, let $K_e$ be the Minkowski sum of $e$ with $D(o)$, where $o$ is the origin, i.e., $K_e$ is the region such that a unit disc $D(q)$ intersects $e$ if and only if $q\in K_e$. The problem of reporting the arcs of $E_C$ intersected by a unit disc 
$D(q)$ is equivalent to reporting the regions of $K = \{ K_e \mid e \in E_C\}$ that contain $q$.
Exploiting the property that the arcs of $E_C$ lie inside a grid cell of diameter~$1$, we show that 
our pseudo-line data structure can be used for reporting the regions of $K$ that contain a query point.

\textbf{\textit{Circular-arc intersection searching.}} 
Our final result is a data structure for the intersection-searching problem 
in which the input objects are arbitrary unit-radius circular arcs rather 
than arcs
 forming the boundary of the union of the unit discs, and the query is a 
 unit disc. 
We present a linear-size data structure with $O(n \log n)$ preprocessing time,  $O(n^{1/2+\eps} + k)$ query time and $O(\log^2 n)$
amortized update time, where $k$ is the size of the output 
and $\eps>0$ is an arbitrarily small, but fixed, constant. This result follows the same approach as earlier data structures for intersection~\cite{AvKO,GJS} and constructs a two-level partition tree. Our main contribution is a simpler characterization of the condition of a unit disc intersecting a unit-radius circular arc.\footnote{We believe the update time can be made worst case by using the 
lazy reconstruction method~\cite{Over83}.}

\paragraph{Road map of the paper.}
The paper is organized as follows: We begin in Section~\ref{sec:dynamic_lower_env}
by describing the dynamic data structure for maintaining the lower envelope of pseudo-lines. Next, we present in Section~\ref{sec:union_maintain} 
the data structure for maintaining the union of unit discs under insertions. Section~\ref{sec:range-search} presents the dynamic data structure for unit-arc intersection searching. Finally, we conclude in Section~\ref{sec:concl} by mentioning a few open problems.

\section{Maintaining Lower Envelope of Pseudo-Lines}
\label{sec:dynamic_lower_env}

We describe a dynamic data structure to maintain the lower envelope of a set of $x$-monotone
pseudo-lines in $\R^2$ under insertions and deletions,
which also works for a more general class of planar curves; see below.

\subsection{Preliminaries}
\label{sec:prelims}

Let $E$ be a family of $x$-monotone pseudo-lines in $\reals^2$; a vertical line crosses each pseudo-line in exactly one point.
Let $\ell$ be a vertical line strictly to the left of 
the first intersection point in $E$. It defines a total order $\leq$ on the pseudo-lines in $E$, 
namely, for $e_1, e_2 \in E$, we have $e_1 \leq e_2$ if and 
only if $e_1$ intersects $\ell$ below $e_2$. Since each 
pair of pseudo-lines in $E$ cross exactly once, it follows 
that if we consider a vertical line $\ell'$ strictly to the 
right of the last intersection point in $E$, the order of 
the intersection points between $\ell'$ and $E$, from 
bottom to top, is reversed.

The \emph{lower envelope} $\cL(E)$ of $E$ is the 
$x$-monotone curve obtained by taking the pointwise 
minimum of the pseudo-lines in $E$, i.e., if we regard each pseudo-line of $E$ as the graph of a univariate function $e(x)$, 
then the lower envelope $\cL(E)$ is the graph of the function $\min_{e\in E} e(x)$, $x\in \reals$. 
A \emph{breakpoint} of $\cL(E)$ is an intersection point of two pseudo-lines that appears on $\cL(E)$, and an \emph{arc} or \emph{segment}
of $\cL(E)$ is the maximal contiguous portion of a pseudo-line of $E$ that appears on $\cL(E)$ (between two consecutive breakpoints).
Combinatorially, $\cL(E)$ can be represented by the sequence of its breakpoints and arcs in the increasing $x$-order;
the first and the last arcs of $\cL(E)$ are unbounded. 
The \emph{upper envelope} $\cU(E)$ of $E$ is similarly the $x$-monotone curve obtained by taking the pointwise maximum of the pseudo-lines in $E$.

In this section, we focus on $\cL(E)$. The following two properties of $\cL(E)$
are crucial for our data structure: 
\begin{itemize}
	\item[(A)] every pseudo-line contributes at most one segment to $\cL(E)$; and  
	\item[(B)] the order of these segments from left to right corresponds exactly to the 
order $\leq$ on $E$ defined above. 
\end{itemize}

We assume a computational model in which primitive operations on pseudo-lines, such as computing the 
intersection point of two pseudo-lines or determining the intersection point of a pseudo-line 
with a vertical line, can be performed in constant time.

\subsection{Data structure and operations}
\label{sec:pseudo-line-DS}

\paragraph{The tree structure.}
Our primary data structure for maintaining $\cL(E)$ is a balanced binary 
search tree (e.g., a red-black tree~\cite{Tar}) $\Xi$, which supports insertion and  deletion operations in $O(\log n)$ 
time.  The leaves of $\Xi$ contain 
the pseudo-lines, sorted from left to right according to the order defined above. 
An internal node $v \in \Xi$ represents the 
lower envelope of the pseudo-lines contained in the subtree rooted at $v$.
More precisely, every leaf $v$ of $\Xi$ stores a single pseudo-line $v.e \in E$. For a node $v$ 
of $\Xi$, we write $v.E$ for the set of pseudo-lines in the subtree rooted at $v$.
We denote the lower envelope of $v.E$ by  $v.\cL$.
Let $w$ (resp. $z$) be the left (resp.\ right) child of $v$. Then $w.\cL$ and $z.\cL$ intersect at one point $v.\chi$, and $v.\cL$ consists of the prefix (resp.\ suffix) of $w.\cL$ until $v.\chi$ (resp.\ of $z.\cL$ from $v.\chi$); see Figure~\ref{fig:children-envelope}.

\begin{figure}[htb]
	\centering
	\includegraphics{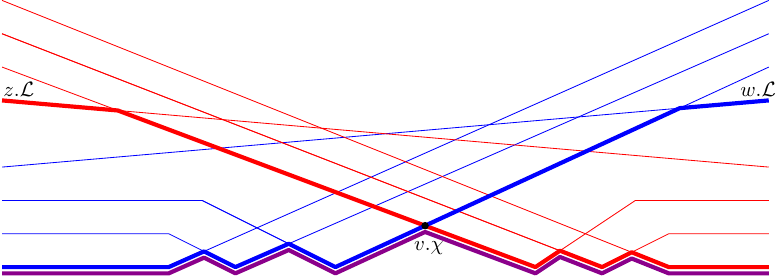}
	\caption{Constructing the lower envelope $v.\cL$ (purple) from $w.\cL$ (red) and $z.\cL$ (blue).}
	\label{fig:children-envelope}
\end{figure}

Each  node $v$ stores the following variables:
\begin{itemize}
\item $f$, $l$, $r$: a pointer to the parent, 
the left child, and the right child of $v$, respectively; $l, r$ are undefined for a leaf, and $f$ is undefined for the root;
\item $\max$: the \emph{last} pseudo-line in $v.E$ (according to
the order defined in Section~\ref{sec:prelims}); 
\item $\chi$: the intersection point of $(v.l).\cL$ and $(v.r).\cL$, the lower envelopes of the left and right children of $v$, if $v$ is an internal node; $\chi$ is undefined for leaves;
\item $\Lambda$: 
	a balanced binary search tree (e.g., a red-black tree) that stores the 
		prefix or the suffix of $v.\cL$, denoted by $\overline{\cL}$, that 
		is not on the lower envelope $(f.v).\cL$; the root of $\Xi$ stores the entire envelope $\cL(E)$.
		The leaves of $\Lambda$ represent the segments of $\overline{\cL}$ sorted from left to right. Each node $\xi$ of $\Lambda$ 
		is associated with a contiguous portion $\overline{\cL}_\xi$ of $\overline{\cL}$. Each leaf $\xi$ stores the endpoints of $\overline{\cL}_\xi$, which consists of a single segment, and the pseudo-line of $v.E$ that supports $\overline{\cL}_\xi$.
		Each inner node $\xi$ of $\Lambda$, with left and right children $\zeta$ and $\eta$, stores the common endpoint $\xi.p$ of $\overline{\cL}_\zeta$ and $\overline{\cL}_\eta$. We note that the two pseudo-lines supporting the last segment of $\overline{\cL}_\zeta$ and the first segment of $\overline{\cL}_\eta$ intersect at $\xi.p$ and the lower envelope of these two pseudo-lines, denoted by $\xi.\lclenv$, represents the lower envelope $v.\cL$ locally in the neighborhood of $\xi.p$. We store these two pseudo-lines at $\xi$. Since $\xi.\lclenv$ can be computed in $O(1)$ time from the two pseudo-lines, for simplicity, we can assume that we also store $\xi.\lclenv$ at $\xi$. See Section~\ref{sec:intersection-point} below for more details on $\Lambda$.
\end{itemize}

During our update procedure, we need to perform split and
join operations on the secondary trees $v.\Lambda$ at various
nodes $v$ in $\Xi$.  Each of these procedures can be implemented in $O(\log n)$ time using the standard methods~\cite[Chapter~4]{Tar}.

\paragraph{Queries.} 
We now describe the two query operations that we perform 
on $\Xi$.

\textbf{\textit{Point-location query.}}
Given a value $x_0 \in \R$, we report the pseudo-line 
$e \in E$ that contains the point on $\cL(E)$ with $x$-coordinate $x_0$. 
Since the root $u$ of $\Xi$ explicitly stores
$\cL(E)$ in a balanced binary search tree $u.\Lambda$,
this query can be answered in $O(\log n)$ time.

\begin{lemma}
\label{lem:vertical_rs}
For a given value $x_0\in\R$, a point-location query can be answered in $O(\log n)$ time.
\end{lemma}

\textbf{\textit{Ray-intersection query.}}
Given a point $q\in\R^2$, we report all pseudo-lines of $E$ that
lie vertically below $q$, i.e., report all pseudo-lines that intersect the ray emanating from $q$ in the $(-y)$-direction. 

Let $q_x$ be the $x$-coordinate of $q$. We perform a point-location query with  
$q_x$ on $\Xi$ and determine the pseudo-line $e$ that contains the point of $\cL(E)$ with $x$-coordinate $q_x$.
If $q$ lies below $e$, we are done.
Otherwise, we store $e$ in the result set and delete $e$ from $\Xi$.  We repeat this step until either $\Xi$ becomes empty or $q$ lies below the lower envelope of the remaining set. Finally, we re-insert all elements from the result set to restore the original set of pseudo-lines.  Overall, we need $k + 1$ point-location queries, $k$ deletions, and $k$ insertions. By Lemma~\ref{lem:vertical_rs}, each point-location query needs $O(\log n)$ time, and below we show that one update operation requires $O(\log^2 n)$ time.  Hence, we obtain the following.

\begin{lemma}
\label{lem:multiVerticalRs}
Let $q \in \R^2$. All $k$ pseudo-lines in $E$ that lie below $q \in \R^2$ can be reported in time $O(\log n + k \log^2 n)$.
\end{lemma}

\paragraph{Updates.} 
To insert or delete a pseudo-line $e$ in $\Xi$, we follow the method of Overmars and van Leeuwen~\cite{Overmars1981}. 
We delete or insert a leaf $z$ for $e$ in $\Xi$ using the standard techniques for balanced binary search trees 
(the $v.\max$ pointers guide the search in $\Xi$)~\cite{Tar}. We update the secondary structure stored at the nodes of $\Xi$, as follows. 
Let $\pi$ be the path in $\Xi$ from the root to $z$. As we go down along $\pi$, for each node $v \in \pi$ and its sibling $w$, 
we construct $v.\cL$ and $w.\cL$ from $(v.f).\cL$, stored as a balanced binary tree.  
If $v$ is the root, then $v$ already stores $v.\cL$, so assume $v$ is not the root and inductively we have $(v.f).\cL$ at our disposal. 
We split $(v.f).\cL$ at $(v.f).\chi$, and let $\Lambda^-$ (resp.\ $\Lambda^+$) be the prefix (resp.\ suffix) of $(v.f).\cL$. 
If $v$ is the left child of $v.f$, then $v.\cL$ (resp.\ $w.\cL$) is obtained by merging $\Lambda^-$ with 
$v.\Lambda$ ($w.\Lambda$ with $\Lambda^+$).  
If $v$ is the right child, then the roles of $v$ and $w$ are reversed. 
When we reach the leaf, we have the lower envelope at the siblings of all non-root nodes in $\pi$. 

After having inserted or deleted $z$, we trace $\pi$ back in a bottom-up manner. When we reach a node $v$, 
we have computed $(v.l).\Lambda$, $(v.r).\Lambda$, and $v.\cL$. At the node $v$, we  first compute the unique intersection point 
$(v.f).\chi$ of $v.\cL$ and $w.\cL$, where $w$ is the sibling of $v$, using the procedure described in the next subsection; 
recall that we already have computed $w.\cL$.  Suppose $v$ is the left child of its parent. 
We split $v.\cL$ into two parts $\cL_v^-, \cL_v^+$ at $(v.f).\chi$, with the former lying to the left. 
Similarly, we split $w.\cL$ into two parts $\cL^-_w, \cL^+_w$ at $(v.f).\chi$ (note that $v.f=w.f$) with the former lying to the left. 
We store $\cL_v^+, \cL_w^-$ as $v.\Lambda$ and $w.\Lambda$, respectively. We also update $v.\max$ and $w.\max$. We then merge $\cL_v^-$ and $\cL_w^+$ to obtain $(v.f).\cL$. We then move to $v.f$. If we reach the root of $\Xi$, then we simply store the envelope at $v.f$ and stop.

Since the height of $\Xi$ is $O(\log n)$, since each split/merge operations takes $O(\log n)$ time, and since, by Lemma~\ref{lem:intersection} below, the intersection point of two envelopes at each node  can be computed in $O(\log n)$ time,
the update procedure takes $O(\log^2 n)$ time. More details can be found, e.g.,
in the original paper by Overmars and van Leeuwen~\cite{Overmars1981}
or in the book by Preparata and Shamos~\cite{PreparataSh85}.

\begin{lemma}
	\label{lem:le-update}
An insert/delete operation in $\Xi$ takes $O(\log^2 n)$ time.
\end{lemma}

\subsection{Finding the intersection point of two lower envelopes}
\label{sec:intersection-point}

Given two lower envelopes $\cL_l$ and $\cL_r$, such that all pseudo-lines 
in $\cL_l$ are smaller than all pseudo-lines in $\cL_r$ and each envelope is stored in a balanced binary tree as described above, we give a procedure to compute the (unique) intersection point $q$ between $\cL_l$ 
and $\cL_r$ in $O(\log n)$ time. In our algorithm, $\cL_l$ 
and $\cL_r$ are stored as balanced 
binary search trees $\Lambda_l$ and $\Lambda_r$. 

The leaves of $\Lambda_l$ 
and $\Lambda_r$ represent the segments 
on the lower envelopes $\cL_l$ and $\cL_r$, sorted from left to
right. 
To ensure that every point on $\cL_l$ and
$\cL_r$ is associated with exactly one leaf of $\Lambda_l$
and $\Lambda_r$, we use the convention that the segments in the
leaves are semi-open, containing their right, but not their left, endpoint in $\Lambda_l$
and their left, but not their right, endpoint in $\Lambda_r$. Recall that we
store \emph{both} endpoints of a segment as well as the pseudo-line supporting the segment at each leaf of
$\Lambda_l$ and $\Lambda_r$,
but the segments are interpreted as relatively semi-open sets,
where the precise endpoint to be included  depends 
on the role that the tree plays in the intersection algorithm.
More concretely, the intersection algorithm uses two items stored  at a leaf $v$ of 
$\Lambda_l$ or $\Lambda_r$:
\begin{itemize}
	\item[(i)] the pseudo-line $v.\cL$ that supports the segment represented by $v$; and 
	\item[(ii)] an endpoint $v.p$ of the segment, namely the left endpoint if $v$ is a leaf of  $\Lambda_l$, 
		and the right endpoint if $v$ is a leaf of $\Lambda_r$.\footnote{If the
segment is unbounded, the endpoint
might not exist. In this case, we use
a symbolic endpoint at infinity that
lies below every other pseudo-line.}
Note that this is exactly the endpoint of the associated
segment that is \emph{not} 
included in the semi-open segment represented by $v$.
This choice is made to ensure a uniform
handling of inner nodes and leaves in the intersection
algorithm.
\end{itemize}

Consider an inner node $v$ of $\Lambda_l$ or $\Lambda_r$.
The intersection algorithm uses the two items stored at $v$: 
\begin{itemize}
	\item[(i)] the lower envelope $v.\lclenv$ of the last (maximum) pseudo-line in the left subtree 
$v.l$ of $v$ and the first (minimum) pseudo-line in the right subtree $v.r$ of $v$; and 
		\item[(ii)] the intersection point $v.p$ of these two pseudo-lines, which is the only breakpoint of $v.\lclenv$. 
\end{itemize}

\begin{figure}[htb]
\centering
	\includegraphics{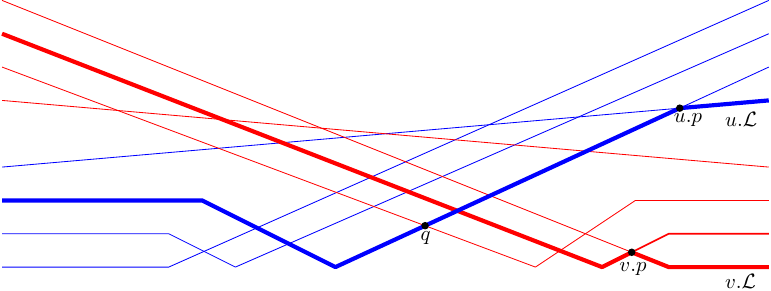}
\caption{An example of Case~1: The pseudo-lines
in $\Lambda_l$ are shown blue, the pseudo-lines
in $\Lambda_r$ are shown red.}
\label{f:case_three_pseudo_lines} 
\end{figure}

\begin{figure}[htb]
\centering
	\includegraphics{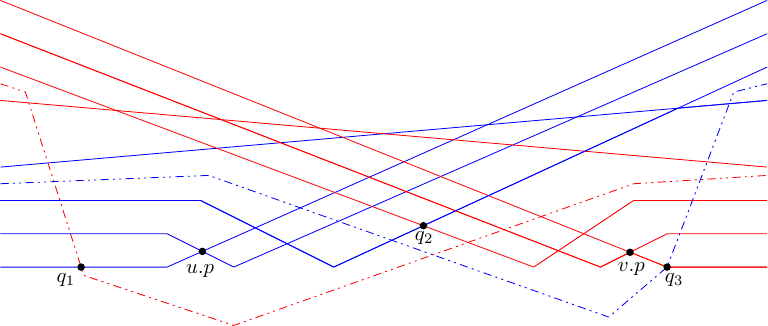}
\caption{An example of Case~3: 
	The pseudo-lines in $\Lambda_l$ (resp.\ $\Lambda_r$) are shown
	blue (resp.\ red).
The solid pseudo-lines are fixed. 
The dashed pseudo-lines are optional, meaning that either
none or exactly one of the dashed pseudo-lines 
is present. The current vertices of the binary search are  $u.p$
and $v.p$,
and Case~3 applies.
Irrespective of the local
situation at $u$ and $v$, the intersection point $q$ of
$\cL_l$ and $\cL_r$
	might be to the left of $u.p$ (e.g., $q_1$ in the figure), between $u.p$ and 
	$v.p$ (e.g., $q_2$ in the figure), or to the right of $v.p$ (e.g., $q_3$ in the figure), depending 
on which one of the dashed pseudo-lines
is present.}
\label{f:hard_case_pseudo_lines} 
\end{figure}

As discussed above, the leaf $u^*$ of $\Lambda_l$ and
the leaf $v^*$ of $\Lambda_r$ whose 
(half-open) segments 
contain the intersection point $q$ between $\cL_l$ and
$\cL_r$ are 
uniquely determined.
Let 
$\pi_l$ be the path in $\Lambda_l$
from the root to $u^*$ and 
$\pi_r$ 
the path in $\Lambda_r$ from the root to $v^*$. 
Our strategy is as follows: we simultaneously 
descend into $\Lambda_l$ and $\Lambda_r$, following the paths $\pi_l$ and $\pi_r$, starting
from the respective roots. 
Let $u$  be the current node in $\pi_l$ 
and let $v$ be the current node in $\pi_r$.
At each step, we perform a local test on 
$u$ and $v$, comparing  $u.p$ with $v.\lclenv$ and $v.p$ with $u.\lclenv$, 
to decide how to proceed. 
The test distinguishes among three possibilities:

\begin{enumerate}
	\item \textit{The point $u.p$ lies on or above the (local) lower envelope $v.\lclenv$.} 
	In this case, $u.p$ lies on or above the envelope $\cL_r$. Therefore, the intersection point $q$ 
	between $\cL_l$ and $\cL_r$ must be equal to or to the left of $u.p$; 
		see Figure~\ref{f:case_three_pseudo_lines}. If $u$ is an inner node, then the desired leaf $u^*$ cannot lie in
		the right subtree $u$ (recall that the half-open segments in the leaves
		of $\Lambda_l$ are considered to be open to the left).
		If $u$ is a leaf, then $u^*$ lies strictly
		to the left of $u$ (recall that in this case, $u.p$ is
		the left endpoint of the segment stored in $u$, so 
		$u.p$ does not belong to the half-open segment in $u$ but to the half-open
		segment in the predecessor-leaf).

	\item \textit{The point $v.p$ lies on or above the (local) lower envelope  $u.\lclenv$.} 
		In this case, $v.p$ lies on or above the entire envelope $\cL_l$, therefore 
		the intersection point $q$ between $\cL_l$ and $\cL_r$
		is equal to or to the right of $v.p$; this situation is symmetric to the one depicted in
		Figure~\ref{f:case_three_pseudo_lines}.
		If $v$ is an inner node, then $v^*$ cannot lie in
		the left subtree $v$ (recall that the segments in the leaves of $\Lambda_r$ are
		considered to be open to the right). If $v$ is a leaf, then $v^*$ lies 
		strictly to the right of $v$ (recall that in this case, $v.p$ is
		the right endpoint of the segment stored in $v$, so $v.p$ does not
		belong to the half-open segment in $v$, but to the half-open segment in the successor-leaf). 

	\item \textit{The point $u.p$ lies below the (local) lower envelope $v.\lclenv$ and the point $v.p$ lies 
		below the (local) lower envelope  $u.\lclenv$:} in this case, the point $u.p$ must lie
		strictly to the left of the point $v.p$. 
		This claim follows from property~(B) of pseudo-lines because  all pseudo-lines
		in $\Lambda_l$ are smaller than all pseudo-lines
		in $\Lambda_r$; see Figure~\ref{f:hard_case_pseudo_lines}.
		Thus, it follows that the intersection point $q$ 
		is strictly to the right of $u.p$ or strictly to the left of
		$v.p$ (both situations can occur simultaneously, if $q$ lies between $u.p$ and $v.p$).
		In the former case, if $u$ is an inner node, then
		$u^*$ lies in $u.r$  or to the right of all leaves in $u.r$, and if $u$ is
		a leaf, then either $u^* = u$ or $u^*$ is a leaf to the right of $u$. In the
		latter case, if $v$ is an inner node, then $v^*$ lies in $v.l$
		or to the left of all leaves in $v.l$, and if $v$ is a leaf, then $v^* = v$
		or $v^*$ is a leaf to the left of $v$.
\end{enumerate}

In the first two cases, it is easy to perform the next
step in the binary search.
In the third case, however, it is not immediately
obvious what to do. The correct choice
might be either to go to $u.r$ or to $v.l$. 
For the straight-line case, Overmars and van Leeuwen
resolve this ambiguity by comparing the slopes
of the relevant lines. For pseudo-lines, however,
there is no notion of slope.
Even worse, it seems that there is no local test
to resolve this situation. For an example,
refer to Figure~\ref{f:hard_case_pseudo_lines}, where
the local situation at $u$ and $v$ does not help to determine
the position of the intersection point $q$. 
We present an alternative strategy, which also applies 
for pseudo-lines.

\begin{figure}
    \centering
    \includegraphics[scale=0.73]{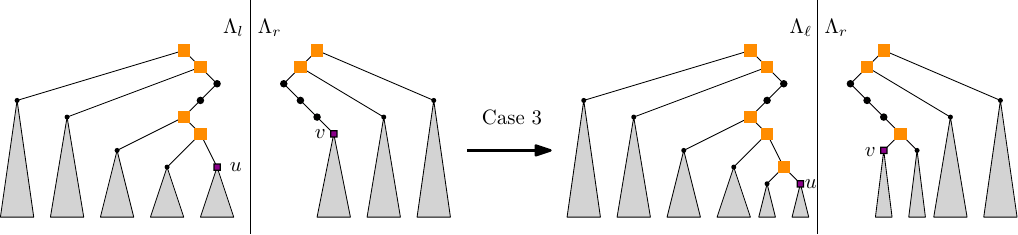}
    \caption{Comparing $u$ to $v$: in Case~3,
    we know that $u^*$ is in $u.r$ or $v^*$ is in $v.l$; we go to
    $u.r$ and to $v.\ell$.}
    \label{fig:case3}
\end{figure}

Throughout the search, we maintain 
the invariant that the subtree at the current node $u$
of $\Lambda_l$ contains the desired leaf $u^*$ \emph{or} the
subtree at the current node $v$ of $\Lambda_r$ contains the desired node $v^*$ 
(or both). In Case~3, as explained above, it holds that
$u^*$ must be
in $u.r$ \emph{or} $v^*$ must be in $v.l$ (or both);
see Figure~\ref{fig:case3}. 
Thus, we will move $u$ to $u.r$
and $v$ to $v.l$. One of these moves must be correct,
but the other move might be mistaken: we might go
to $u.r$ even though $u^*$ is in $u.l$; or to $v.l$
even though $v^*$ is in $v.r$. To account for this possible 
mistake,
we remember the current node $u$ in a stack \texttt{uStack} and
the current node $v$ in a stack \texttt{vStack}. Then, if
it becomes necessary, we can backtrack and
revisit the other subtree $u.l$ or $v.r$. This approach leads
to the general situation shown in Figure~\ref{fig:invariant}:
The desired leaf $u^*$ is in the subtree of $u$ or in a left subtree of a node
on $\texttt{uStack}$, while the desired leaf $v^*$ is in the 
subtree  of $v$ or in a right 
subtree of a node on $\texttt{vStack}$, and at least one of
$u^*$ or $v^*$ must be in the subtree of $u$ or of $v$, respectively. Now, if
Case~1 occurs when comparing $u$ to $v$, we can exclude the
possibility that $u^*$ is in $u.r$. Thus, $u^*$ might be in
$u.l$, or in the left subtree of a node in \texttt{uStack}; 
see Figure~\ref{fig:case1}.
To make progress, we now compare $u'$, the top of \texttt{uStack},
with $v$. Again, one of the three cases occurs:
\begin{itemize}
	\item[(i)] In Case~1, we can deduce that going to $u'.r$ was mistaken, and we move
$u$  to $u'.l$, while $v$ does not move. 
	\item[(ii)] In the other cases, we cannot rule out that $u^*$ is to the right of $u'$, and we 
move $u$ to $u.l$, keeping the invariant that $u^*$ is either
below $u$ or in the left subtree of a node on \texttt{uStack}.
However, to ensure that the search progresses, we now must also
move $v$:
		\begin{itemize}
			\item In Case~2, we can rule out that $v^*$ lies in $v.l$, and we move 
$v$ to $v.r$. 
				\item In Case~3, we move $v$ to $v.l$. 
		\end{itemize}
\end{itemize}

	In this way, we keep the invariant and always make progress: in each step,
we either discover at least one new node on either of the two correct search
paths, or we pop one erroneous move from one of the two stacks.
Since the total length of the correct search paths is
$O(\log n)$, and since we push a new element onto the stack 
only when discovering a new node on either of the correct search
paths, 
the total search time is $O(\log n)$; see 
Figures~\ref{f:intersection_point_demo2} and~\ref{f:intersection_point_demo}  and
Table~\ref{fig:algodemo} in Appendix for 
an example run of the algorithm.

The following pseudo-code gives the details
of our algorithm, including all corner cases.

{\SetAlgoNoLine%
\begin{algorithm}[H]
oneStep($u$, $v$)\\
\Indp     do compare($u$, $v$):\\
\Indp         Case 3: \\
\Indp           \If{$u$ is not a leaf}
                {
                    uStack.push($u$); $u \leftarrow u.r$
                }
                \If{$v$ is not a leaf}
                {
                    vStack.push($v$); $v \leftarrow v.l$
                }
                \If{$u$ and $v$ are leaves}
                {
                    return $u = u^*$ and $v = v^*$
                }
\Indm       Case 1:\\
\Indp           \uIf{uStack is empty}
                {
                    $u \leftarrow u.l$\;
                }
                \uElseIf{ $u$ is a leaf}
                {
                    $u \leftarrow \texttt{uStack.pop().}l$
                }
                \Else
                {
                    $u' \leftarrow \texttt{uStack.top()}$\\
                    do compare($u'$, $v$)\\
\Indp                   Case 1: \\
\Indp                      uStack.pop(); $u \leftarrow u'.l$\;
\Indm                   Case 2:\\
\Indp                       $u \leftarrow u.l$ \\
                            \If{$v$ is not a leaf}
                            {
                                $v. \leftarrow v.r$
                            }
\Indm                  Case 3:\\
\Indp                       $u \leftarrow u.l$\\
                            \If{ $v$ is not a leaf}
                            {
                                vStack.push($v$); $v\; \leftarrow v.l$
                            }
                             
                }
\Indm       Case 2:\\
\Indp           symmetric
\end{algorithm}
}

\begin{figure}
    \centering
    \includegraphics{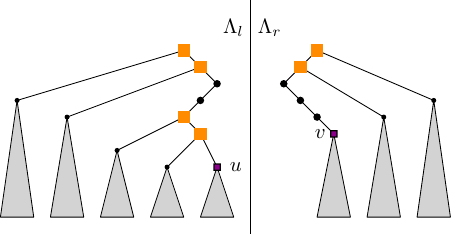}
    \caption{The invariant:
    the current search nodes are $u$ and $v$.
    \texttt{uStack} contains all nodes on the
    path from the root to $u$ where the path goes to a right
    child (orange squares), \texttt{vStack} contains all
    nodes from the root to $v$ where the path goes to a left child 
    (orange squares). The final leaves $u^*$ and $v^*$ are in one of the
    gray subtrees; and at least one of them is under $u$ or under $v$.}
    \label{fig:invariant}
\end{figure}
\begin{figure}
    \centering
    \includegraphics[scale=0.73]{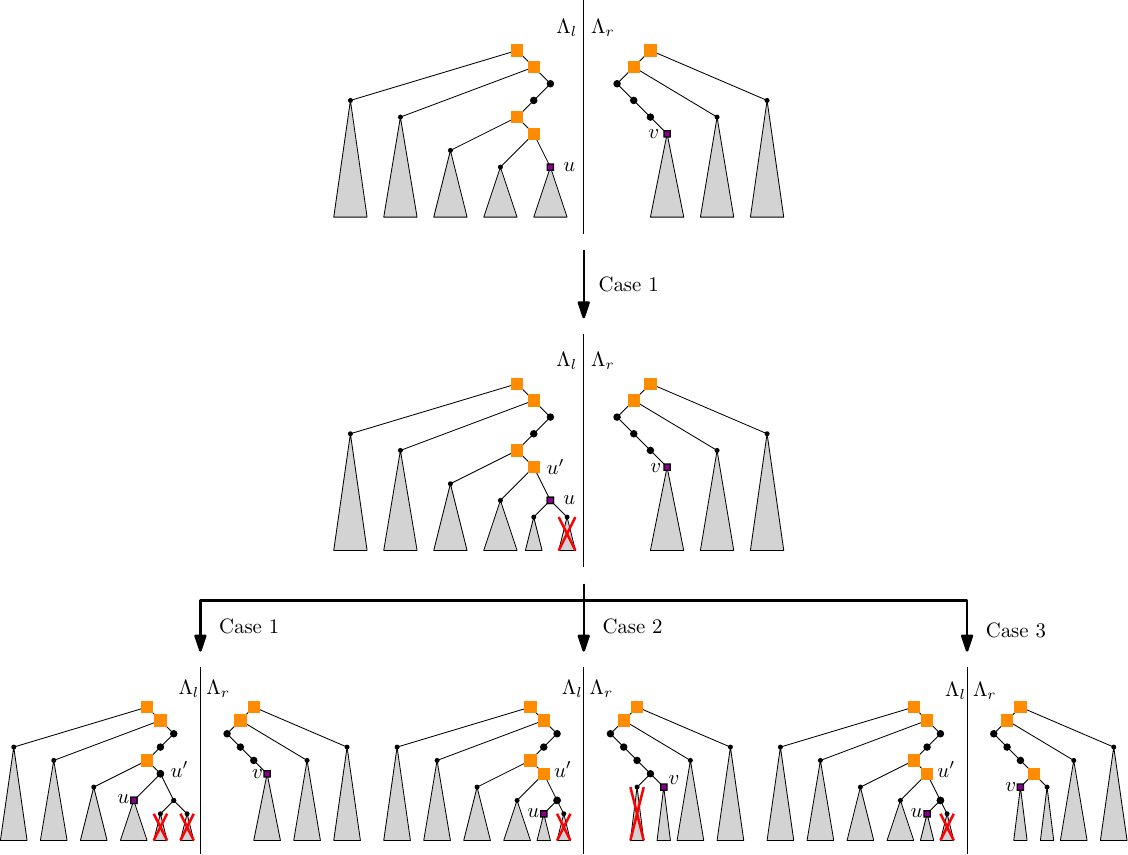}
    \caption{Comparing $u$ to $v$:
    in Case~1, we know that
    $u^*$ cannot be in $u.r$. 
    We compare $u'$ and $v$ to decide
    how to proceed: 
    in Case~1, we know that $u^*$ cannot be in $u'.r$; we go to
    $u'.l$; in Case~2, we know that $u^*$ cannot be in $u.r$ and that 
    $v^*$ cannot be in $v.l$; we go to $u.l$ and to $v.r$; in Case~3,
    we know that $u^*$ is in $u'.r$ (and hence in $u.l$) or in $v.l$;
    we go to $u.l$ and to $v.l$. Case~2 is not
    shown, as it is symmetric.}
    \label{fig:case1}
\end{figure}

We will show that the search procedure maintains
the following invariant:
\begin{invariant}
\label{inv:intersection}
The leaves in all subtrees $u'.l$, for 
$u' \in \textup{\texttt{uStack}}$, together with the
leaves under $u$ constitute a 
prefix of the leaves in  $\Lambda_l$. This prefix 
contains $u^*$. 
Similarly, the leaves in all subtrees
$v'.r$, $v' \in \textup{\texttt{vStack}}$,  
together with the leaves under $v$ constitute a 
contiguous suffix of the leaves of $\Lambda_r$.
This suffix contains $v^*$. Furthermore, we have
$u \in \pi_l$ or $v \in \pi_r$ (or both).
\end{invariant}

Invariant~\ref{inv:intersection} holds at the 
beginning, when both stacks are empty,
$u$ is the root of $\Lambda_l$ and $v$ is the
root of $\Lambda_r$. To show that the invariant
is maintained, we first consider the special case
when one of the two searches has already discovered
the correct leaf. Recall that $\pi_l, \pi_r$ are the root-to-leaf paths to $u^*$, $v^*$ in 
$\Lambda_l$ and $\Lambda_r$, respectively.

\begin{lemma}
\label{lem:leaf_inv}
Suppose that Invariant~\ref{inv:intersection} holds and
that Case~3 occurs when comparing $u$ to $v$.
If $u = u^*$, then
$v \in \pi_r$ and, if $v$ is not a leaf, then $v.l \in \pi_r$.
Similarly, if $v = v^*$, then $u \in \pi_l$ and,
if $u$ is not a leaf, then $u.r \in \pi_l$.
\end{lemma}

\begin{proof}
We consider the case $u = u^*$; the other case is symmetric. 
Let $e_u$ be the segment of $\cL_l$ stored in $u$.
By Case~3, the point $u.p$ is strictly to the left of the point $v.p$. 
Furthermore, since $u = u^*$, the intersection point
$q$ lies on $e_u$. Thus, $q$ cannot
be to the right of $v.p$, because otherwise
$v.p$ would be a point on $\cL_r$ that lies below $e_u$
and to the left of $q$, which is impossible.
Since $q$ is strictly to the left of $v.p$,
Invariant~\ref{inv:intersection} shows that if $v$ is
an inner node, then $v^*$ must be in $v.l$ (and hence
both $v$ and $v.l$ lie on $\pi_r$), and if $v$ is a
leaf, then $v = v^*$.
\end{proof}

We can now show that the invariant is maintained.

\begin{lemma}
Procedure \textup{\texttt{oneStep}} either correctly
reports the desired leaves $u^*$ and $v^*$, or maintains 
Invariant~\ref{inv:intersection}. In the latter case, either
it pops an element from one of the two stacks, or it
discovers a new node on $\pi_\ell$ or $\pi_r$.
\end{lemma}

\begin{proof}
First, suppose Case~3 occurs. The invariant
that \texttt{uStack} and $u$ cover a prefix of
$\cL_l$ and that \texttt{vStack} and $v$
cover a suffix of $\cL_r$ is maintained.
Furthermore, if both $u$ and $v$ are inner nodes,
Case~3 ensures that $u^*$ is in $u.r$ or
to the right of $u$, or that $v^*$ is in 
$v.l$ or to the left of $v$. Suppose the former case
holds. Then, Invariant~\ref{inv:intersection} 
implies that $u^*$ must be in $u.r$, and
hence $u$ and $u.r$ lie on $\pi_l$. 
Similarly, in the second case, 
Invariant~\ref{inv:intersection} gives that
$v$ and $v.l$ lie on $\pi_r$.
Thus,
Invariant~\ref{inv:intersection} is maintained
and we discover a new node on $\pi_l$ or
on $\pi_r$.
Now, assume $u$ is a leaf and $v$ is an inner node.
If $u \neq u^*$, then as above, 
Invariant~\ref{inv:intersection} and Case~3 imply
that $v \in \pi_r$ and $v.l \in \pi_r$, 
and the lemma holds.
If $u = u^*$, the lemma follows from Lemma~\ref{lem:leaf_inv}.
The case that $u$ is an inner node and $v$ a 
leaf
is symmetric. If both $u$ and $v$ are leaves, Lemma~\ref{lem:leaf_inv}
implies that \texttt{oneStep} correctly  reports $u^*$ and
$v^*$.

Second, suppose Case~1 occurs. Then, 
$u^*$ cannot be in $u.r$, if $u$ is an
inner node, or  $u^*$ must be to the left
of a segment left of $u$, if $u$ is 
a leaf.
Now, if \texttt{uStack}
is empty, Invariant~\ref{inv:intersection}
and Case~1 imply that $u$ cannot be a leaf 
(because $u^*$ must be in the subtree of $u$)
and that $u.l$ is a new node on $\pi_l$.
Thus, the lemma holds in this case. 
Next, if $u$ is a leaf,  
Invariant~\ref{inv:intersection} and
Case~1 imply that $v \in \pi_r$. Thus, we pop
\texttt{uStack} and maintain the invariant; 
the lemma holds.
Now, assume that \texttt{uStack} is not
empty and that $u$ is not a leaf. 
Let $u'$ be the top of $\texttt{uStack}$.
First, if the comparison between $u'$ and $v$ results
in Case~1, then $u^*$ cannot be in
$u'.r$, and in particular, $u \not\in \pi_l$.
Invariant~\ref{inv:intersection} shows
that $v \in \pi_r$, 
and we pop an element from \texttt{uStack},
so the lemma holds.
Second, if the comparison between $u'$ and $v$
results in Case~2, then $v^*$ cannot
be in  $v.l$, if $v$ is an inner node.
Also, if $u \in \pi_l$, then necessarily also
$u.l \in \pi_l$, since Case~1
occurred between $u$ and $v$. If $v \in \pi_r$,  
since Case~2 occurred between $u'$ and $v$, the node
$v$ cannot
be a leaf and $v.r \in \pi_r$. Thus, in either case
the invariant is maintained and we discover a new
node on $\pi_l$ or on $\pi_r$.
Third, assume the comparison between
$u'$ and $v$ results in Case~3. If
$u \in \pi_l$, then 
also $u.l \in \pi_l$,
because $u.r \in \pi_l$ was excluded by
the comparison between $u$ and $v$. In this case,
the lemma holds. If $u \not\in \pi_l$,
then also $u'.r \not \in \pi_l$, so the fact
that Case~3 occurred between $u'$ and $v$ implies that
$v.l$ must be on $\pi_r$ (in this case,
$v$ cannot be a leaf, since otherwise we would
have $v^* = v$ and Lemma~\ref{lem:leaf_inv} would
give $u'.r \in \pi_l$, which we have already ruled out).
The argument for Case~2 is symmetric.
\end{proof}

The following lemma finally shows that our intersection
procedure finds the desired point in logarithmic time.

\begin{lemma}
	\label{lem:intersection}
    The intersection point $q$ between $\cL_l$ and 
    $\cL_r$ can be computed in $O(\log n)$ time.
\end{lemma}

\begin{proof}
    In each step, we either discover a new node of 
    $\pi_l$ or of $\pi_r$, or we pop an element
    from \texttt{uStack} or \texttt{vStack}. 
    Elements are pushed only when
    at least one new  node on $\pi_l$ or 
    $\pi_r$ is discovered. 
    As $\pi_l$ and $\pi_r$ are each a path from the root to a leaf 
    in a balanced binary tree, 
    we  need $O(\log n)$
    steps.
\end{proof}

Putting everything together, we obtain the following:

\begin{theorem}
	\label{th:pseudo-lines}
	A set $E$ of $n$ pseudo-lines can be maintained in a data structure so that (i) a pseudo-line can be inserted/deleted in $O(\log^2 n)$ time; 
	(ii) for a query value $x_0\in\R$, the point of $\cL(E)$ with the $x$-coordinate $x_0$ and the input pseudo-line containing this point 
	can be computed in $O(\log n)$ time; and 
	(iii) all $k$ pseudo-lines of $E$ lying below a query point $q\in\R^2$ can be reported in $O(\log n+k\log^2 n)$ time. 
\end{theorem}

\section{Maintaining the Union of Unit Discs under Insertions}
\label{sec:union_maintain}

Let $S$ be a set of $n$ points in $\reals^2$, and let $U := U(S) = \bigcup_{p\in S} D(p)$ be the union of the unit discs centered at the points of $S$.  
In this section, we describe a data structure  that
maintains $\edges$, the set of edges of $\bd U$. After the insertion of a new point to $S$, it updates $\edges$ in $O(\log n + k\log^2 n)$ time, where $k$ is the number of changes (insertions plus deletions) in the set $\edges$. It can also report, within the same time bound, the area of $U$, denoted
$\area U$, after the insertion of each point. 

This section is organized as follows. Section~\ref{sec:overview} gives a high-level description of the overall data structure and of the update procedure. Section~\ref{sec:edge-intersection} describes the data structure for reporting the set of edges in $\edges$ that intersect a unit disc, which relies on the data structure described in the previous section. Finally, Section~\ref{sec:properties} proves a few key properties of $\edges$ that are crucial for our data structure.

\subsection{Overview of the data structure}
\label{sec:overview}

The overall data structure consists of two parts.
Let $\GG$ be a uniform grid in $\R^2$ such that the diameter of each grid cell is $1$. We call a grid cell of $\GG$ \emph{active} if it intersects $U$. Let $\G \subset \GG$ be the set of active grid cells. Each unit disc intersects $O(1)$ grid cells, so $|\G| = O(n)$. 
We define the \emph{key} of a grid cell to be the $x$- and $y$-coordinates of 
its bottom left corner, and we induce a total ordering on the grid cells by using the lexicographic ordering on their keys. 
Using this total ordering, we store $\G$ in a balanced binary search tree (e.g.\ red-black tree) $\Omega$. A membership query and an update operation on $\G$ can be performed in $O(\log n)$ time~\cite{Tar}.

\begin{figure}[ht]
	\centering
	\includegraphics{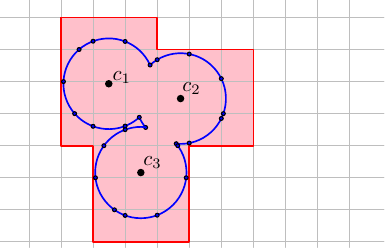}
	\caption{The grid imposed over the union of unit discs.  The active cells are highlighted in pale red.}
    	\label{f:grid_structure} 
\end{figure}

We overlay $U$ with $\GG$. If an edge of $\bd U$ intersects 
more than one grid cell, then we split it at the 
boundaries of the cells that it crosses (see Figure~\ref{f:grid_structure}). We can therefore assume that each edge of $\edges$ lies within a single cell. For each cell $C\in\G$, let $\edges_C \subseteq \edges$ denote the set of edges that lie inside $C$. 
We maintain $\edges_C$ in a dynamic data structure $\Psi_C$, described in Section~\ref{sec:edge-intersection},  that 
\begin{itemize}
	\item[(i)] for a query point $q$, reports, in $O(\log n +k_C\log^2 n)$ time, the subset $\edges_{q,C} \subseteq \edges_C$ of edges that intersect $D(q)$, where $k_C=|\edges_{q,C}|$; and 
	\item[(ii)]  
can handle insertion or deletion of an edge in $\edges_C$ in $O(\log^2 n)$ time.
\end{itemize}
See Lemma~\ref{lem:edge-intersection} below.

Using $\Omega$ and $\Psi_C$, for all grid cells $C\in\G$, the insertion of a point $q$ into $S$ is handled as follows. To avoid confusion, we use $U$ (resp.\ $U^{\mathrm{new}}$) to denote $U(S)$ immediately before (resp.\ after) the insertion of $q$.
\begin{enumerate}
\item\label{step:one} Compute the set $\GG_q$ of $O(1)$ grid cells that the new disc $D(q)$ intersects. 
\item\label{step:two} Find the subset $\G_q = \GG_q \cap \G$ of active cells (before the insertion of $q$) that intersect $D(q)$.
\item\label{step:edge-intersection} For each cell $C\in \G_q$, using the data structure $\Psi_C$, report the subset $\edges_{q,C} \subseteq \edges_C$ of edges that $D(q)$ intersects. 
	Set $\edges_q = \bigcup_{C\in\G_q} \edges_{q,C}$ and $k=|\edges_q|$.
\item\label{step:update} Compute the set $I_q$ of new edges on $U^{\mathrm{new}}$. We split the edges of $I_q$ at the grid boundaries so that each edge lies within one grid cell. For each cell $C\in\GG_q$, let $I_{q,C} \subseteq I_q$ be the set of edges that lie inside $C$.
\item\label{step:update-DS} For each cell $C\in\GG_q$, delete the edges of $\edges_{q,C}$ from $\Psi_C$ and insert the edges of $I_{q,C}$ into $\Psi_C$. If $C \not\in \G$, insert $C$ into $\Omega$.
\item\label{step:area} Compute $\area U^{\mathrm{new}}$.
\end{enumerate}

Steps~\ref{step:one} and~\ref{step:two} are straightforward and can be implemented in $O(\log n)$ time using $\Omega$. 
Steps~\ref{step:edge-intersection} and~\ref{step:update-DS} can be implemented using the procedures described in 
Section~\ref{sec:edge-intersection}. By Lemma~\ref{lem:edge-intersection}, the total time spent in these two steps\footnote{Wherever we describe a certain data structure, the parameter $k$ pertains to the output size of a query in that specific data structure.} is 
$O(\log n + (k+|I_q|)\log^2 n)=O((k+1)\log^2 n)$ because as we will see below, $|I_q|=O(k+1)$. We now describe how to compute the set $I_q$ of new edges (Step~\ref{step:update}) and $\area U^{\mathrm{new}}$ (Step~\ref{step:area}).

\paragraph{Updating the boundary of the union.}
First, consider the case when $k=0$, then either $D(q) \subset U$ or $D(q)\cap U = \emptyset$. Since the diameter of each grid cell in $\GG$ is $1$, at least one of the grid cells, denoted by $\omega$, is 
fully contained in $D(q)$. If $\omega \in \G_q$, then
$\omega\subset U$ (because $\omega\cap U\ne \emptyset$ but $\omega\cap\bd U=\emptyset$ since $\edges_q=\emptyset$)
and therefore $D(q) \subset U$; otherwise $D(q) \cap U = \emptyset$. If $D(q) \subset U$, then $\bd U^{\mathrm{new}}=\bd U$, and there is nothing to do. On the other hand, if $D(q)\cap U=\emptyset$, then the entire $\bd D(q)$ appears on 
$\bd U^{\mathrm{new}}$. We split $\bd D(q)$ at the boundary of the grid cells, and $I_q$ is the resulting set of edges.

We now assume that $k > 0$. The set $I_q$ contains two types of edges:
\begin{itemize}
	\item[(i)] The edges that lie on the boundaries of older discs. These edges are the portions of the edges of $\edges_q$ that lie outside $D(q)$.
	\item[(ii)] The edges that lie on $\bd D(q)$. These edges are (maximal) connected arcs of $\bd D(q)\setminus U$.
\end{itemize}
To compute the first type of edges, for each edge $e\in \edges_q$, we compute the intersection points of $e\cap \bd D(q)$.
If $e \not\subset D(q)$ then we compute $e \setminus D(q)$, which comprises one or two arcs, and which we  add to $I_q$. 

Let $V$ be the set of intersection points of $\bd D(q)$ and the edges of $\edges_q$. We sort $V$ along $\bd D(q)$. These intersection points partition $\bd D(q)$ into arcs. Each such arc $\gamma$ either lies inside $U$ or outside $U$, and we can detect it in $O(1)$ time. If $\gamma$ lies outside $U$, we add $\gamma$ to $I_q$. 

It follows from the above discussion that $|I_q|=O(k+1)$ and that the total time spent in computing $I(q)$ is $O((k+1)\log n)$.
\begin{figure}[htb]
	\centering
	\includegraphics{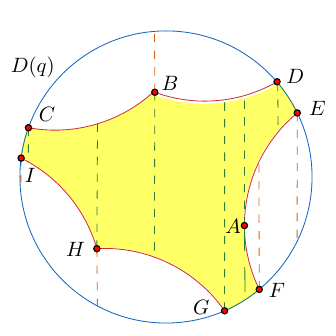}
	\caption{Arrangement $\A$ and its vertical decomposition $\A^\nabla$; shaded pseudo-trapezoids do not lie in $U$.}
	\label{fig:arrangement}
\end{figure}

\paragraph{Computing the area of the union.} We now describe how we extend the procedure for computing $I_q$ to compute $\Delta A = \area U^{\mathrm{new}}-\area U$.
If $k=0$, then as described above, we determine whether $D(q)\subset U$ or $D(q)\cap U=\emptyset$. We have $\Delta A=0$ in the former case, and $\Delta A = \area D(q)=\pi$ in the latter case. We now focus on the case $k>0$.

Let $\edges_q^{\mathrm{in}} = \{ e\cap D(q) \mid e \in \edges_q \}$ be the portions of edges in $\edges_q$ clipped within $D(q)$. $\edges_q^{\mathrm{in}}$ can be computed, in $O(k)$ time, by adapting the procedure for computing $I_q$. 
We note that the relative interiors of edges in $\edges_q^{\mathrm{in}}$ are pairwise disjoint. Let 
$\A$ be the arrangement of $\edges_q^{\mathrm{in}}\cup \bd D(q)$ within $D(q)$, i.e., the decomposition of $D(q)$ induced by these arcs (see~\cite{hs-a-18} for further details on geometric arrangements).
Let $\A^\nabla$ be the vertical decomposition of $\A$, i.e., the refinement of $\A$ obtained  by drawing vertical rays in both $+y$- and $-y$-directions from each vertex of $\A$ or a point of vertical tangency on an arc of $\edges_q^{\mathrm{in}}$, within $D(q)$, until it meets another edge of $\A$. $\A^\nabla$ partitions the faces of $\A$ into pseudo-trapezoids, each bounded by at most two vertical segments and by two circular 
arcs at top and bottom.  See Figure~\ref{fig:arrangement}.
Each pseudo-trapezoid $\tau\in\A^\nabla$ is either contained in $U$ or disjoint from $U$. 
Let $\F$ be the set of faces of $\A^\nabla$ that do not lie in $U$. Then $\Delta A = \sum_{\tau\in \F} \area\tau$.

$\A^\nabla$ can be computed in $O(k\log k)$ time using a sweep-line algorithm~\cite{dBCvKO} (recall that $k>0$ here). 
The sweep-line algorithm can be adapted in a straight-forward manner to compute $\F$ within the same time bound.
For each face $\tau \in\F$, we compute $\area\tau$ in $O(1)$ time and then add them up to compute $\Delta A$. Finally, we set 
$\area U^{\mathrm{new}}=\area U+\Delta A$. After having computed the edges of $U^{\mathrm{new}}$, the total time spent in computing 
$\area U^{\mathrm{new}}$ is $O(k\log k)$.

Putting everything together, we obtain the following.

\begin{theorem}
\label{thm:Near_Optimal_Overhead}
(i) The  boundary edges of the union of a set of $n$ unit discs can be maintained under insertion 
	in a data structure of $O(n)$ size so that a new disc can be inserted in 
	$O((k+1)\log^2 n)$ time, where $k$ is the total number of changes on the boundary of the union.

(ii) The same data structure can also report the area of the union after the insertion of each disc in 
	$O((k+1)\log^2 n)$ time, where $k$, as above, is the total number of changes on the boundary of the union.
\end{theorem}

\subsection{Edge intersection data structure}
\label{sec:edge-intersection}

Let $C$ be an axis-parallel square with diameter $1$, representing a grid cell of $\GG$, and let $\edges_C$ be the set of edges of $\bd U$ that lie in $C$. 
We describe a dynamic data structure $\Psi_C$ to report the edges of $\edges_C$ intersected by a unit disc. It can quickly insert or delete an edge of $\edges_C$.

\begin{figure}[htb]
\centering
\includegraphics{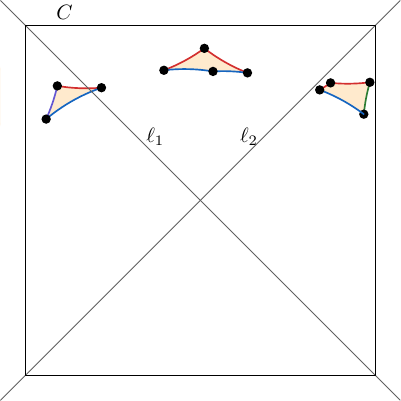}
	\caption{An example of $\partial U \cap C$ with four types of edges $E_t$ (red), $E_b$ (blue), $E_l$ (magenta), and $E_r$ (green).
	}
\label{f:E_t_example} 
\end{figure}

Let $\ell_1$ and $\ell_2$ be the lines that support the diagonals of $C$. The lines $\ell_1$ and
$\ell_2$ divide the plane into four quadrants: the top quadrant $Q_t$, the right
quadrant $Q_r$, the bottom quadrant $Q_b$, and the left quadrant $Q_l$.
We partition the edges of $\edges_C$ into four sets $E_t$, $E_r$, $E_b$, and $E_l$, 
depending on which quadrant $Q_t$, $Q_r$, $Q_b$, or $Q_l$
contains the center of the respective disc; see Figure~\ref{f:E_t_example}.
If a disc center lies on a dividing line $\ell_1$, $\ell_2$,
then the tie is broken in favor of $E_t$ or of $E_b$ (in that order).
We focus on the edges of $E_t$; analogous statements hold also for the other three edge sets.\footnote{%
	The algorithm in De Berg~et~al.~\cite{BBJW21} also draws a uniform grid so the diameter of each grid cell is $1$. 
	For each grid cell, they build data structures for the unit discs whose centers lie in that cell. In contrast, for each grid cell, we build data structures on union arcs that lie in that cell.}
Before describing the data structure, we prove a few key properties of $E_t$.

\begin{lemma}
\label{l:open_lower_semi_circle}
Each edge in $E_t$ is a portion of a lower semi-circle.
\end{lemma}

\begin{proof}
Let $e \in E_t$, and let $c \in Q_t$ be the center of the unit disc whose boundary contains $e$. 
Since the cell $C$ has unit diameter, $c$  lies outside $C$
and above the line that supports the top side of $C$.
Thus, $e$, which lies inside $C$, is a portion of the open lower semi-circle of $D(c)$.
\end{proof}

\begin{lemma}
\label{lemma:disoint_x_projection}
The $x$-projections of the (relative interiors of the) 
edges in $E_t$ are pairwise disjoint.
\end{lemma}

\begin{proof}
Let $e_i$ and $e_j$ be two distinct edges of $E_t$. 
Suppose that there is a vertical line $\ell$ that 
intersects both $e_i$ and $e_j$, in points $p_i$ and $p_j$,
respectively. For concreteness, assume
that $p_i$ lies below $p_j$.
By Lemma~\ref{l:open_lower_semi_circle}, the point $p_i$ 
lies on the lower semi-circle of a disc $D_i$ whose center is above 
the upper side of $C$. 
This means that the vertical segment that connects $p_i$ 
to the upper side of $C$ is fully contained in $D_i$. 
But then, $p_j$ cannot be on the boundary $\partial U$ of $U$.
Thus, the vertical line $\ell$ cannot exist, and the $x$-projections
of the edges in $E_t$ have pairwise disjoint interiors.
\end{proof}

By Lemma~\ref{lemma:disoint_x_projection}, 
the edges in $E_t$ can be ordered from left to right,
according to their $x$-projections.
We number them
as  $e_1, \dots, e_m$, according to this order.

\begin{figure}[htb]
	\centering
	\includegraphics{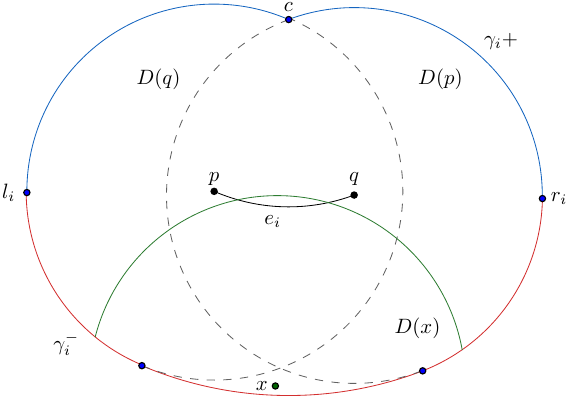}
	\caption{An arc $e_i$ of $E_t$ with endpoints $p$ and $q$, its Minkowski sum with a unit disc, and its upper (blue) and lower arcs (red) $\gamma^+$ and $\gamma^-$ respectively. A unit disc $D(q)$ intersects $e_i$ if and only if $q$ lies above $\gamma_i^-$ and below $\gamma_i^+$.}
	\label{fig:Mink}
\end{figure}

For an edge $e_i \in E_t$, let $K_i = e_i \oplus D(0)$ be the Minkowski sum of $e_i$ with the unit disc centered at the origin, i.e., 
\[ K_i = \{ a\in\reals^2 \mid \exists b\in e_i\, \|a-b\| \le 1 \} .\]
It is easily seen that a unit disc $D(q)$ intersects $e_i$ if and only if $q\in K_i$. See Figure~\ref{fig:Mink}.
We divide $\bd K_i$ at its leftmost and rightmost points $l_i, r_i$
into an \emph{upper arc} $\gamma_i^+$ and a \emph{lower arc} $\gamma_i^-$. 
We say that a point $x$ lies \emph{above} (resp. \emph{below}) an arc if the vertical ray emanating from $x$ in $+y$-direction (resp. $-y$-direction) intersects the arc. The following observation is straightforward:

\begin{lemma}
	\label{lem:mink}
	\begin{itemize}
		\item[(i)] The unit disc $D(q)$ intersects $e_i$ if and only if $q$ lies below $\gamma_i^+$ and above $\gamma_i^-$.
		\item[(ii)] Let $c$ be the center of the unit disc whose boundary contains $e_i$, and let $p$ and $q$ be the endpoints of $e_i$.
			The upper arc $\gamma_i^+$ is the upper envelope of the upper semi-circles of $D(p)$ and $D(q)$. 
			The lower arc $\gamma_i^-$ consists of portions of the lower semi-circles of $D(p)$, $D(q)$, and the disc 
			of radius $2$ centered at $c$.
	\end{itemize}
\end{lemma}

Set
$\Gamma^+ = \{\gamma_i^+ \mid e_i \in E_t\}$ and
$\Gamma^- = \{\gamma_i^- \mid e_i \in E_t\}$.
The following lemma states three crucial properties
of the arcs in $\Gamma^+$ and $\Gamma^-$.

\begin{lemma}
\label{lem:properties}
The arcs in $\Gamma^+$ and $\Gamma^-$ have the 
following properties.
\begin{description}
	\item[(P1)] Let $e_i, e_j \in E_t$ be two distinct edges with $i<j$. Then the lower arcs $\gamma_i^-$ and $\gamma_j^-$ 
		intersect and cross in exactly one point, and $\gamma_i$ (resp.\ $\gamma_j$) appears on $\cL(\{\gamma_i^-,\gamma_j^-\})$
		only before (resp.\ after) their intersection point.
    \item[(P2)] Let $e_i$, $e_j$, and $e_h$ be three 
    edges in $E_t$ with $i < j < h$, and let $q \in \R^2$. If $q$ lies 
    below $\gamma_i^+$ and $\gamma^+_h$, then $q$ also lies below $\gamma^+_j$. 
    Furthermore, the upper semi-circle of the disc centered at every endpoint of an edge in $E_t$ appears on the upper envelope $\cU(\Gamma^+)$, 
    and the order of arcs on $\cU(\Gamma^+)$ corresponds to the $x$-order of the endpoints of the edges of $E_t$.
    \item[(P3)] For every vertical line $\ell$, all intersection 
    points of $\ell$ with the arcs in $\Gamma^-$ lie below all 
    intersection points of $\ell$ with the arcs in $\Gamma^+$.
\end{description}
\end{lemma}
Lemma~\ref{lem:properties} is proved in Section~\ref{sec:properties}. In the remainder of this section, we describe the edge-intersection data structure and the query procedure assuming Lemma~\ref{lem:properties} holds.

The data structure $\Psi_C$ for $E_t$ consists of two parts:  $\Delta^+$ and $\Delta^-$ that
dynamically maintain the sets $\Gamma^+$ and $\Gamma^-$, respectively.
The purpose of $\Delta^+$ (resp. $\Delta^-$) is to efficiently answer 
the following query: given a point $q \in \R^2$, report the upper 
(resp.~lower) arcs that are above (resp.~below) $q$. 
Additionally, $\Delta^+$ has the property that it returns the answer incrementally,
one arc at a time on demand. Both $\Delta^+, \Delta^-$ support insertion/deletion of arcs.

\begin{figure}[htb]
        \centering
	\includegraphics{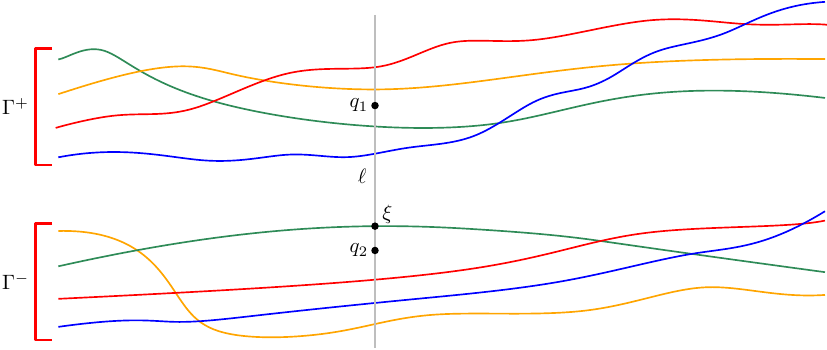}
    \caption{
    Illustration of the search procedure. There are four pairs 
    of upper and lower arcs (each pair has a distinct  color), and two query points $q_1$ and $q_2$ lying on a vertical line
    $\ell$; $\xi$ is the intersection point of $\ell$ with $\cU(\Gamma^-)$. The point $q_1$ (resp. $q_2$) lies above (resp.\ below)
	$\xi$, so it suffices to search in $\Gamma^+$ (resp.\ $\Gamma^-$).
}
    	\label{f:Section_3_Overview} 
\end{figure}

\paragraph{Answering an edge-intersection query.}
With $\Delta^+, \Delta^-$ at our disposal, the edges of $E_t$ that intersect a unit disc $D(q)$ are reported as follows. By Lemma~\ref{lem:mink}~(i), 
we wish to report the edges $e_i$ such that $\gamma_i^-$ lies below $q$ and $\gamma_i^+$ lies above $q$.  Here is the basic idea: 
Let $\ell$ be the vertical line passing through $q$. Assume, for the sake of exposition, that we know
the intersection point $\xi$ between $\ell$ and 
the upper envelope of $\Gamma^-$.
If $q$ lies above $\xi$ (e.g., $q_1$ in Figure~\ref{f:Section_3_Overview}), then
$q$ lies above all the lower arcs that cross
$\ell$. Therefore  it suffices to search  the structure
$\Delta^+$ to report the upper arcs $\gamma_i^+$ that lie above
$q$---in this case, $e_i$ intersects $D(q)$ if and only if $\gamma_i^+$ lies above $q$.
On the other hand, if the center $q$ coincides
with or lies below $\xi$ (e.g., $q_2$ in Figure~\ref{f:Section_3_Overview}), then by property (P3), $q$ lies below all upper arcs that $\ell$ intersects. We therefore search
$\Delta^-$ to report the lower arcs that lie
below $q$---in this case, $e_i$ intersects $D(q)$ if and only if $\gamma_i^-$ lies below $q$.

Unfortunately, we cannot easily maintain $\cU(\Gamma^-)$ and thus cannot compute the point $\xi$.
Hence, the query procedure is a bit more involved: 
We use $\Delta^+$ to return the upper arcs that lie above $q$ incrementally, one by one, on demand. 
For each upper arc $\gamma_i^+$ reported by $\Delta^+$, we check in $O(1)$ time whether 
$e_i$ intersects $D(q)$. If so, 
we add $e_i$ to the output list and ask $\Delta^+$ to report the next upper arc that lies above $q$. If all the upper arcs 
above $q$ turn out to be induced by edges of $E_t$ that intersect $D(q)$, 
we output this list of edges as the desired set $\edges_q$ and stop. Indeed,
if all the reported edges from the query of $\Delta^+$ intersect 
$D(q)$, then we can conclude that $q$ is above $\xi$ and this is 
the complete answer. 

On the other hand, if we detect that the edge $e_j\in E_t$ corresponding to an upper 
arc $\gamma_j^+$ reported by $\Delta^+$ does not intersect 
$D(q)$, then by Lemma~\ref{lem:mink}~(i), $q$ lies below $\gamma_i^-$ and thus below $\xi$.
As argued above, in this case,
we will obtain the full result by querying $\Delta^-$.  We note that an edge of $E_t$ intersecting $D(q)$ is reported at most twice.

As we will see below, if $|\edges_q| =k$ then searching in $\Delta^-, \Delta^+$ takes $O(\log n+k\log^2 n)$ and $O(\log n+k)$ time, respectively. Hence, the total query time is $O(\log n + k\log^2 n)$.
We now describe the data structures $\Delta^-$ and $\Delta^+$.

\paragraph{Maintaining the lower arcs.}

We extend each lower arc $\gamma_i^-$ into a bi-infinite curve $\overline{\gamma}_i^-$ by adding a ray of very large positive (resp. negative) slope in the $+y$-direction from its right (resp. left) endpoint; see Figure~\ref{fig:pseudo-line}. We choose the same slope for all rays and their  values are chosen infinitely large so that any query point $q$ lies below the extended curve if and only if it lies below the original lower arc $\gamma_i^-$. (We should regard this extension as symbolic in the sense that we do not choose any specific value of the slope of these rays.)
Let $\overline{\Gamma}^-$ denote the resulting set of bi-infinite curves.

\begin{figure}[htb]
	\includegraphics{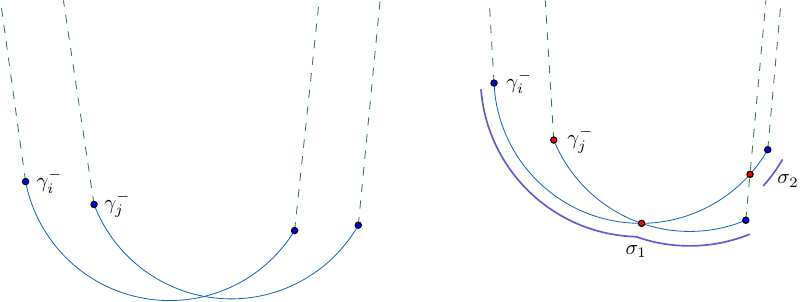}
	\caption{(a) Extending each lower arc to a bi-infinite curve. (b) Illustration of the proof of Lemma~\ref{lem:lower-arcs-PL}.}
	\label{fig:pseudo-line}
\end{figure}

\begin{lemma}
	\label{lem:lower-arcs-PL}
	$\overline{\Gamma}^-$ is a set of pseudo-lines. Furthermore, $\cL(\overline{\Gamma}^-)$ is the same as $\cL(\Gamma^-)$ except the two extension rays that appear as unbounded segments at each end of the lower envelope.
\end{lemma}
\begin{proof}
	Suppose there are two extended curves $\overline{\gamma}_i^-, \overline{\gamma}_j^- \in \overline{\Gamma}^-$,  with $i<j$, that intersect at two points 
	$\sigma_1, \sigma_2$. By (P1), one of these intersection points, say, $\sigma_1$, is the intersection point of 
	$\gamma_i^-$ and $\gamma_j^-$, and $\gamma_j^-$ appears below $\gamma_i^-$ to the right of $\sigma_1$. 
	Suppose $\sigma_2$ lies to the right of $\sigma_1$; see Figure~\ref{fig:pseudo-line}. 
	Then $\sigma_2$ is the intersection point of $\gamma_i^-$ and the extension ray from the right endpoint of $\gamma_j^-$. 
	But this would imply that the lower arc $\gamma_i^-$ extends beyond the right endpoint of the arc $\gamma_j^-$ and thus 
	$\gamma_i^-$ appears on $\cL(\{\gamma_i^-,\gamma_j^-\})$ to the right of $\sigma_1$, contradicting (P1). 

	A similar contradiction arises if $\sigma_2$ lies to the left of $\sigma_1$. Hence, $\overline{\gamma}_i^-, \overline{\gamma}_j^-$ 
	intersect in exactly one point and the two arcs cross at that point, which implies that $\overline{\Gamma}^-$ is a set of 
	pseudo-lines. 
	The second part of the lemma now follows from the fact that every pair of arcs in $\Gamma^-$ intersects.
\end{proof}

In view of Lemma~\ref{lem:lower-arcs-PL}, we can use the data structure
described in Section~\ref{sec:dynamic_lower_env} to report the lower arcs that lie below a query point $q\in\reals^2$.
Recall that this data structure uses linear space, answers a ray-intersection query in $O(\log n + k\log^2 n)$ time, where $k$ is the output size, 
and handles insertion/deletion of a curve in $O(\log^2 n)$ time.

\textbf{Remark.} After we insert a new unit disc, we update the set $E_t$, which leads to deleting or inserting many lower arcs.
In order to ensure that Property~(P1)
hold at all times, we first delete all the old lower arcs 
from $\Delta^-$ and then insert the new ones. 

\paragraph{Maintaining the upper arcs.}

Let $P = \langle p_1 <_x p_2 <_x \cdots <_x \dots <_x p_r\rangle$ 
be the sequence of endpoints of the edges of $E_t$.
To keep the structure simple, if two edges of $E_t$ meet at 
a single point, we keep only one copy of that point, but for each point $p_i$, we remember 
the upper arcs incident to $p_i$,
from left to right. 
For a point $p_i \in P$, let $\ucircle(p_i)$ denote the upper semi-circle of $D(p_i)$.
By Lemma~\ref{lem:mink}~(ii), $\cU(\Gamma^+)$ is the upper envelope of $\{ \ucircle(p_i) \mid 1 \le i\le r\}$, 
and by property (P2), each point in $P$ contributes an arc $s_i$ to $\U(\Gamma^+)$.

The arcs $\ucircle(p_i)$ can be extended to pseudo-lines (similar to lower arcs), and we can construct the pseudo-line data structure of 
Section~\ref{sec:dynamic_lower_env} to maintain $\cU(\Gamma^+)$. Here, we describe a simpler and slightly more efficient data structure.
The data structure $\Delta^+$ is a red-black tree.
The $i$th leftmost leaf stores the point $p_i$ of $P$. For a leaf $v$, we will use $p_v$ to denote the point stored at $v$.
Each leaf also stores pointers \texttt{rn} and \texttt{ln} to the 
right and the left neighboring leaf, respectively, if they exist. 
Each internal node stores a pointer \texttt{lml} to the leftmost leaf 
of its right subtree. Each subtree of $\Delta^+$ corresponds to a contiguous portion of $\cU(\Gamma^+)$.
For each $p_i$, we store the centers of the (at most two) unit discs whose boundaries contain $p_i$.

\textbf{\textit{Query.}}
Let $q$ be a query point, and let $\ell$ be the vertical line passing through $q$.  Recall that the structure $\Delta^+$ reports the upper arcs lying above a 
query point $q$ incrementally, one arc at a time. Here is the outline of the query procedure: We first find the arc $s_q$ of $\cU (\Gamma^+)$ intersected by $\ell$. If $q$ lies above $s_q$, then $q$ does not lie below any upper arc in $\Gamma^+$, and we stop. So assume that $q$ lies below $s_q$. 
Let $p_i\in P$ be the point such that $\ucircle(p_i)$ contains  $s_q$.
We report the upper arc(s) corresponding to $p_i$. Next, we  traverse the sequence $P$ 
starting at $p_i$, going both to the right and to the left, and reporting the upper arcs corresponding to these points until we find a point $p_j$ (in each direction) such that $q$ lies above $\ucircle(p_j)$. By (P2), if $\ucircle(p_j)$ for $j>i$ (resp. $j<i$) lies below $q$ then so does $\ucircle(p_h)$ for all $h>j$ (resp. $h<j$), so we can stop. We now describe how we perform these steps efficiently using $\Delta^+$. 

By following a path from the root, we first find the leaf $v$ storing the point $p_v$
such that the $s_q$ lies on $\ucircle(p_v)$. The search down the tree is carried out as follows:
Suppose we are at an internal node $u$. We compute the breakpoint $\sigma_u$ of $\cU(\Gamma^+)$ 
that separates the portions of $\cU(\Gamma^+)$ represented by the left and right subtrees of $u$:
We use the pointer \texttt{lml}($u$) to obtain $w$, the leftmost leaf in the right subtree of $u$. Using
\texttt{ln}($w$), we find the predecessor of $p_w$.
The breakpoint $\sigma_u$ is the intersection point of $\ucircle(p_w)$ and $\ucircle(p_{\texttt{ln}(w)})$.
If $x(\sigma_u) \le x(q)$, we visit the right child of $u$; otherwise we visit the left child of $u$.

Once we reach the leaf $v$ such that the arc $s_q$ of $\cU(\Gamma^+)$ lies on $\ucircle(p_v)$, we 
traverse the leaves of $\Delta^+$, starting at $v$ and 
going both to the right and to the left, say, we first go right and then left. 
Suppose we are going right and we are at a leaf $u$. When we are asked to report an upper arc, we test whether $q$ lies below 
$\ucircle(p_u)$.  If the answer is yes, then we report the (at most two) arcs of $\Gamma^+$ corresponding to $p_u$, and move to the next leaf in the right direction. If $u$ was the rightmost leaf or $q$ lies above $\ucircle(p_u)$, we start visiting left starting from $\texttt{ln}(v)$ and do the same as above. Once we visited the leftmost leaf $u$ or detected that $q$ lies above $\ucircle(p_u)$, then we stop and declare that all upper arcs of $\Gamma^+$ lying above $q$ have been reported.

The correctness of the query procedure follows from property (P2) of upper arcs.
If the procedure reported a total of $k$ arcs then the time taken by the procedure is $O(\log n + k)$.

\textbf{\textit{Update.}} An upper arc may be inserted into or deleted from $\Delta^+$ in $O(\log n)$ time by
simply removing the endpoints of the deleted arc and inserting the endpoints of the new arc into $\Delta^+$ and updating the auxiliary information stored at the node of $\Delta^+$.

Building similar data structures for $E_b, E_l$, and $E_r$ and putting everything together, we obtain the main result of this subsection.
\begin{lemma}
	\label{lem:edge-intersection} 
	The edges of $\bd U$ lying inside a grid cell of $\GG$ can be maintained in a linear-size data structure so 
	that all edges intersecting a unit disc can be reported in $O((k+1)\log^2 n)$ time, where $k$ is 
	the number of reported arcs. The data structure can be updated in $O(\log^2 n)$ time per 
	insertion/deletion.
\end{lemma}

\subsection{Proof of Lemma~\ref{lem:properties}}
\label{sec:properties}

We now prove properties (P1)--(P3) of upper and lower arcs stated in Lemma~\ref{lem:properties}.
We begin with Property~(P1).

\begin{figure}[ht]
\centering
\includegraphics{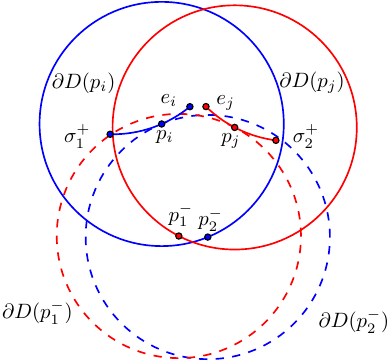}
\caption{Illustration of the proof that $\gamma_i^-$ and $\gamma^-_j$ 
intersect exactly once (see Lemma~\ref{l:lower_curves_intersect_once}).}
\label{f:lower_curves_intersect_once} 
\end{figure}

\begin{lemma}
\label{l:lower_curves_intersect_once}
Let $e_i, e_j \in E_t$ be two distinct edges with $i<j$. Then the lower arcs $\gamma_i^-$ 
	and $\gamma^-_j$ intersect in exactly one point and the two arcs cross at that point. Furthermore, $\gamma_i^-$ (resp.\ $\gamma_j^-$) appears on $\cL(\{\gamma_i^-,\gamma_j^-\})$ only before (resp. after) their intersection point.
\end{lemma}
\begin{proof}
First, we observe that if $e_i$ and $e_j$ have a common
endpoint $q$, then $q$ does not 
contribute to $\cL(\{\gamma_i^-, \gamma^-_j\})$, 
i.e., the unit disc $D(q)$ lies strictly 
above $\cL(\{\gamma_i^-, \gamma^-_j\})$. 
Indeed, assume for a contradiction that 
$\cL(\{\gamma_i^-, \gamma^-_j\})$ contains a point $r$
with distance $1$ from $q$.
Then, the unit disc $D(r)$ 
is tangent to $e_i$ and $e_j$ at the point $q$. 
However, this is impossible, since $e_i$ and $e_j$ belong to the 
lower semi-circles of two distinct unit discs.
    
Next, notice that the arcs $\gamma_i^-$ and $\gamma^-_j$ 
intersect at least once since any 
two points in the grid cell $C$ have distance at most $1$.
Suppose $p^-_1$ is a point on $\gamma_j^-$
and $p^-_2$ a point on $\gamma_i^-$, with 
$p^-_1 <_x p^-_2$;
refer to Figure~\ref{f:lower_curves_intersect_once}.
Assume for a contradiction that
both $p^-_1$ and $p^-_2$ appear on 
	$\cL(\{\gamma_i^-,  \gamma^-_j\})$.
Consider the upper semi-circle $\sigma_1^+$ of 
$D(p_1^-)$ and the 
upper semi-circle $\sigma_2^+$ of 
$D(p_2^-)$. The upper semi-circle $\sigma_1^+$ touches
$e_j$ in a point $p_j$, and the upper semi-circle $\sigma_2^+$
touches $\partial D(p_1^-)$ in a point $p_i$.
	Since $p_1, p_2$ lie on the lower semicircles of $\bd D(p_j), \bd D(p_i)$, respectively, and $\|p_i-p_j\| \le 1$, 
	the upper semi-circles $\sigma^+_1$ and $\sigma^+_2$ intersect 
exactly once, and since $p^-_1 <_x p^-_2$, the semi-circle $\sigma^+_1$ 
appears to the left of $\sigma^+_2$ on the upper envelope $\U$ of 
$\sigma^+_1$ and $\sigma^+_2$. 
The point $p_i$ must be on $\U$, since otherwise $p^-_1$ 
would be inside $D(p_i)$, which contradicts the fact that $p^-_1$ 
belongs to $\cL(\{\gamma_i^-,  \gamma^-_j\})$; 
and similarly for $p_j$.
This implies that $p_i \geq_x p_j$. 
Now, recall that the common endpoint of $e_i$ and $e_j$
(if it exists) does 
not contribute to $\cL(\{\gamma_i^-,  \gamma^-_j\})$,
so we actually  have $p_i >_x p_j$,
which contradicts the assumption $i < j$.
Thus, it follows that $\gamma_i^-$ and $\gamma_j^-$ intersect
	exactly once and that $\gamma_i^-$ appears before $\gamma_j^-$ 
	on $\cL(\{\gamma_i^-,  \gamma^-_j\})$. Furthermore, this argument also implies that $\gamma_i^-$ lies above 
	$\gamma_j^-$ after their intersection point and vice-versa. We can conclude that
	$\gamma_i^-$ (resp., $\gamma_j^-$) appears on $\cL(\{\gamma_i^-,  \gamma^-_j\})$
	only before (resp., after) their intersection point.
\end{proof}

\begin{figure}[htb]
        \centering
        \includegraphics{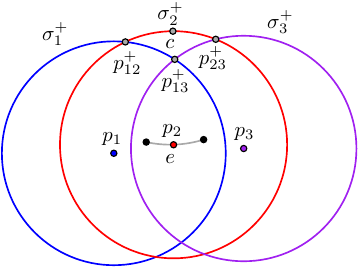}
    \caption{Illustration of the proof of Lemma~\ref{l:upper_curves_order}.}
    	\label{f:upper_curves_order}
\end{figure}

\begin{lemma}
Let $e_i$, $e_j$, and $e_h$, for $i < j < h$, be
three edges in $E_t$, and let $q \in \R^2$ be a point. 
If $q$ is below $\gamma_i^+$ and $\gamma^+_h$ 
then $q$ is also below $\gamma^+_j$.
Furthermore, for every endpoint $p$ of every edge of $E_t$, 
	the upper semi-circle of $D(p)$ appears on $\U(\Gamma^+)$. 
The $x$-order of the arcs on $\U(\Gamma^+)$ corresponds to the $x$-order 
of the endpoints of the edges of $E_t$. 
\label{l:upper_curves_order}
\end{lemma}

\begin{proof}
Let $p_1$, $p_2$, and $p_3$ be points on three distinct edges 
of $E_t$, with $p_1 <_x p_2 <_x p_3$; 
see Figure~\ref{f:upper_curves_order}.
Let $\sigma^+_1$, $\sigma^+_2$, and $\sigma^+_3$ be the upper 
semi-circles of $D(p_1)$, $D(p_2)$, 
and $D(p_3)$, respectively. 
Let $p^+_{12}$ and $p^+_{23}$ be the 
intersection points $\sigma^+_1 \cap \sigma^+_2$ and 
$\sigma^+_2 \cap \sigma^+_3$, respectively. Note that 
these intersection points exist, since the distance between 
any two points in the grid cell $C$ is at most 1. 
Since $p_1 <_x p_2$, we have that $\sigma^+_1$ 
appears to the left of $\sigma^+_2$ on $\U(\{\sigma^+_1,\sigma^+_2\})$. 
Let $c$ be the center of the circle containing the edge $e$ of $E_t$ that contains $p_2$. 
The point $c$ is on $\sigma^+_2$, since $p_2$ belongs to a 
lower semi-circle of radius 1. 
Moreover, $c$ is not below $\sigma^+_1$, since otherwise 
we would have that $p_1$ is in the interior of $D(c)$, 
contradicting the fact that $p_1$ 
lies on an edge of $E_t$. This means that $p^+_{12} \leq_x c$. 
The same argument implies that $p^+_{23} \geq_x c$ and therefore 
$p^+_{12} \leq_x p^+_{23}$. This in turn implies that the intersection 
point, $p^+_{13}$, between $\sigma^+_1$ and $\sigma^+_3$ is below or on  
$\sigma^+_2$ and therefore every point that lies below $\sigma^+_1$ and 
$\sigma^+_3$ also lies below $\sigma^+_2$. 

The lemma readily follows from the above observations.
\end{proof}

Next, we show that for any two distinct edges 
$e_i, e_j\in E_t$,  the upper arc $\gamma_i^+$ and the
lower arc $\gamma^-_j$ are 
disjoint. Furthermore, we show that $\gamma_i^+$ is above $\gamma^-_j$, 
hence proving Property~(P3).

\begin{figure}[htb]
\centering
\includegraphics{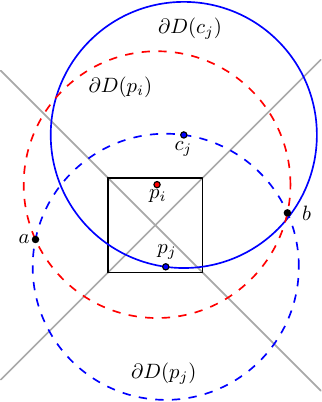}
\caption{Illustration of the proof that 
$\gamma_i^-$ and $\gamma^+_j$ do 
not intersect (Lemma~\ref{lem:upper_above_lower}).}
\label{f:upper_lower_curves_not_intersect}
\end{figure}

\begin{lemma}
\label{lem:upper_above_lower}
Let $e_i$ and $e_j$ be two distinct edges in $E_t$, and let 
$\ell$ be a vertical line. If $\ell$ intersects with $\gamma_i^+$
and $\gamma^-_j$ in the points $p$ and $q$, respectively, 
then, $p >_y q$. 
\end{lemma}

\begin{proof}
First, we show that $\gamma_i^-$ and $\gamma^+_j$ do 
not intersect. Suppose for a contradiction that 
$\gamma_i^-$ and $\gamma^+_j$ intersect at a point $a$. 
This means that $a$ is one of the intersection points of 
$\partial D(p_i)$ and $\partial D(p_j)$, where $p_i \in e_i$ and 
$p_j \in e_j$. Then $a \leq_y p_i$ and
$a \geq_y p_j$,
since $a$ lies on $\gamma_i^-$ and on $\gamma^+_j$. 
The same argument applies to the second intersection point, $b$, 
between $\partial D(p_i)$ and $\partial D(p_j)$. Assume that $a <_x b$, 
which implies that $a$ and $b$ belong to $Q_l$ and $Q_r$, 
respectively. Let $c_j$ be the center point of $e_j$. 
The point $c_j$ lies on the upper semi-circle of $D(p_j)$, and 
it belongs to $Q_t$ which means that $c_j \in D(p_i)$. This means 
that $p_i \in D(c_j)$ which contradicts the fact that $e_i$ 
belongs to $E_t$; see Figure~\ref{f:upper_lower_curves_not_intersect}.
    
Since $\gamma_i^+$ and $\gamma^-_j$ do not intersect, it must be that
that in the common $x$-interval of $\gamma_i^+$ and $\gamma^-_j$,
the arc $\gamma_i^+$ is strictly above or below $\gamma^-_j$. 
The edge $e_i$ is above $\gamma^-_j$ and therefore $\gamma_i^+$ is 
above $\gamma^-_j$.
\end{proof}

\section{Intersection Searching of Unit Arcs with a Unit Disc}
\label{sec:range-search}

In this section, we address the following 
intersection-searching problem: Preprocess a collection 
$\C$ of circular arcs of unit radius into a data structure so that for a 
query point $x \in \R^2$, 
the arcs in $\C$ intersecting the unit disc $D(x)$ can be reported efficiently.
By splitting each arc of $\C$ into at most three arcs, we can ensure that each arc of $\C$ lies in the lower or upper semi-circle of its supporting disc.
We assume for simplicity that every arc in $\C$ belongs 
to the lower semi-circle of its supporting disc. A similar data structure can be constructed for arcs lying in upper semi-circles.

Let $e \in \C$ be a unit-radius circular arc whose center is at $c$, and let $p_1$ and $p_2$ 
be its endpoints. A unit disc $D(x)$ intersects $e$ if and only 
if $e \oplus D(0)$, the Minkowski sum of $e$ with a unit disc at
the origin, contains the center $x$.
Let $z = D(p_1) \cup D(p_2)$, and let $D^+(c)$ be the disc of radius 
$2$ centered at $c$;
$z$ divides $D^+(c)$ into three regions 
(see Figure~\ref{f:zones_and_line_intersect_q}):
(i) $z^+$, the portion of $D^+(c)\setminus z$ above $z$, 
(ii) $z$ itself, and (iii) $z^-$, the portion of  
$D^+(c)\setminus z$ below $z$. 
It can be verified that $e \oplus D(0) = z \cup z^-$. 
We give an alternate characterization of $z\cup z^-$, 
which will help in developing the data structure.

Let $\ell$ be a line that passes through the tangent 
points, $p'_1$ and $p'_2$, of $D(p_1)$ and $D(p_2)$ 
with $D^+(c)$, respectively, and let $\ell^-$ be the 
halfplane below $\ell$. Set $L(e) = D^+(c)\cap \ell^-$.

\begin{lemma}
\label{l:line_intersect_q}
If $\partial D(p_1)$ and $\partial D(p_2)$ intersect at 
two points (one of which is always $c$),
then $\ell$ passes through 
$q = (\partial D(p_1)\cap \partial D(p_2)) \setminus \{c\}$. 
Otherwise, $c \in \ell$.
\end{lemma}

\begin{proof}
Assume that $q$ exists. The quadrilateral 
$(c, p_1, q, p_2)$ is a rhombus, since all its edges have length $1$. 
Let $\alpha$ be the angle $\angle p_1qp_2$ and $\beta$ be the 
angle $\angle cp_1q$.  The angle $\angle qp_1p'_1$ is equal to $\alpha$, 
since the segment $(c,p'_1)$ is a diameter of $D(p_1)$. 
The angle $\angle p_1qp'_1$ is equal to $\frac{\beta}{2}$, 
since $\triangle p_1qp'_1$ is an isosceles triangle. 
The same arguments apply to the angle $\angle p_2qp'_2$, 
implying that the angle $\angle p'_1qp'_2$ is equal to $\pi$.

Assume now that $q$ does not exist. Then the segment $(p_1,p_2)$ 
is a diameter of $D(c)$. The segment $(c,p'_1)$ is a 
diameter of $D(p_1)$. The segment $(p_1,p_2)$ coincides with 
$(c,p'_1)$  at the segment $(c,p_1)$. 
The same argument applies to the segment $(c,p'_2)$, implying 
that the angle $\angle p'_1qp'_2$ is equal to $\pi$.
\end{proof}

\begin{figure}[ht]
	\centering
	 \begin{subfigure}{0.4\textwidth}
        \includegraphics{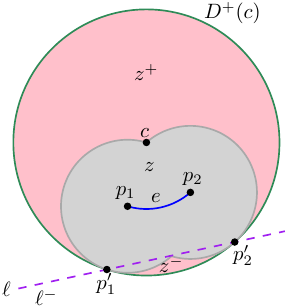}
        \captionsetup{justification=centering}
    	\label{f:zones}
    \end{subfigure}
    \hspace{3em}
    \begin{subfigure}{0.4\textwidth}
        \includegraphics{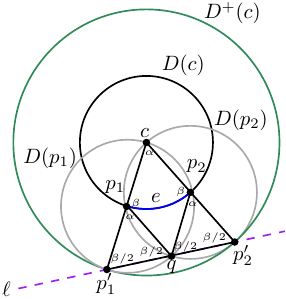}
        \captionsetup{justification=centering}
    	\label{f:line_intersect_q} 
	\end{subfigure}
    \caption{(On the left) Partition of $D^+(c)$ into three regions: 
    $z^+$, $z$ and $z^-$. (On the right) Illustration 
    of Lemma~\ref{l:line_intersect_q}.}
    \label{f:zones_and_line_intersect_q}
\end{figure}

The following corollary summarizes the criteria 
for the intersection of a unit circular arc with a unit disc.

\begin{corollary}
\label{c:criteria_unit_arc}
Let $e$ be a circular arc in $\C$ with endpoints $p_1$ and $p_2$ 
and center $c$. Then
$z\cup z^- = z \cup L(e)$. Furthermore,
$e$ intersects a unit disc $D(x)$ if and only 
if at least one of the following conditions 
is satisfied: (i) $x \in D(p_1)$ (or $p_1 \in D(x)$),  
(ii) $x \in D(p_2)$ (or $p_2 \in D(x)$), and 
(iii) $x \in L(e)$.
\end{corollary}

We thus construct three separate data structures. 
The first data structure preprocesses the left endpoints of
the arcs in $\C$ for unit-disc range searching, the 
second data structure preprocesses the right endpoints of arcs 
in $\C$ for unit-disc range searching, and the third data structure 
preprocesses $\cL = \{L(e) \mid e \in \C\}$ for inverse 
range searching, i.e., reporting all regions in $\cL$ that 
contain a query point. Using standard disk range-searching data 
structures (see e.g.~\cite{Agarwal1994,Agarwal2017}), we can build 
these three data structures so that each of them takes $O(n)$ space
and answers a query in $O(n^{1/2+\eps}+k)$ time, where $k$ is the 
output size. Furthermore, these data structures can handle 
insertions/deletions in $O(\log^2 n)$ time using the lazy-partial-rebuilding technique~\cite{Over83}. We 
conclude the following:

\begin{theorem}
Let $\C$ be a set of $n$ unit-circle arcs in $\R^2$. Then,
$\C$ can be preprocessed into a data structure of linear 
size so that for a query unit disc $D$, all arcs of $\C$ 
intersecting $D$ can be reported in 
$O(n^{1/2+\eps} + k)$ time,
where $\eps$ is an arbitrarily small constant and $k$ is the output size.
Furthermore the data structure can be updated under 
insertion/deletion of a unit-circle arc
in $O(\log^2 n)$ amortized time. 
\end{theorem}

\section{Conclusion}
\label{sec:concl}

In this paper, we presented linear-size dynamic data structures for 
maintaining the union of unit discs under insertions and for maintaining the lower envelope of pseudo-lines in the plane under both insertions and deletions. We also presented a linear-size structure for storing a  set of circular arcs of unit radius 
(not necessarily on the boundary of the union of the corresponding discs), so that all input arcs intersecting a query unit 
disc can be reported quickly. We conclude by mentioning a few open problems:
\begin{itemize}
	\item[(i)] Can the boundary of the union of unit discs be maintained in an output-sensitive manner when we allow both insertions and deletions of discs? The challenge in extending our data structure to handle the deletion of a unit disc $D$ is to quickly report the 
		new edges of $\U$ that lie in the interior of $D$.
	\item[(ii)] Can our data structure be extended to maintain the boundary of the union of discs of arbitrary radii? Although the union boundary still has linear size, many of the structure properties of the union used by our data structure no longer hold, e.g., two upper (or lower) arcs may intersect at two points and they cannot be treated as pseudo-lines.
	\item[(iii)] Can the area of the union of unit discs be maintained under insertions and deletions in sublinear time per update? 
		As mentioned in the introduction, it is not clear how to extend 
		Chan's data structure~\cite{Chan20a} for maintaining the volume of the convex hull of a 
		set of points in $\R^3$ to our setting.
\end{itemize}

\paragraph*{Acknowledgments.}
We thank Haim Kaplan and Micha Sharir for helpful discussions, and reviewers for their useful comments.

\bibliographystyle{abbrv}
\bibliography{sample-base}

\section*{Appendix: An example run of the algorithm in Section~\ref{sec:intersection-point}}
In this appendix we illustrate the algorithm, described in Section~\ref{sec:intersection-point}, for computing the 
intersection point of the lower envelopes of two sets of pseudo-lines such that all lines in one set are smaller than all in the other set.

\begin{figure}[htb]
    \centering
        \centering
	\includegraphics[width=0.6\textwidth]{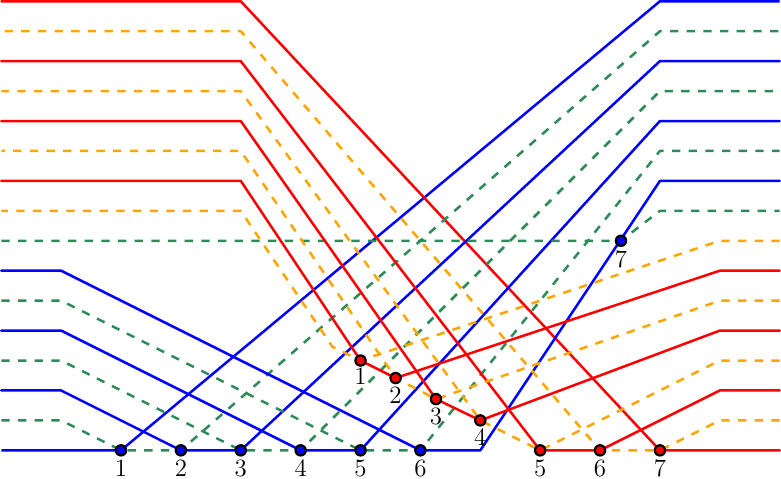}
        \caption{Two sets of pseudo-lines and their lower envelopes: 
	(i) the blue and green pseudo-lines, 
	(ii) the red and orange pseudo-lines. The blue and the 
	red dots represent the vertices on the lower envelopes.}
    	\label{f:intersection_point_demo2} 
\end{figure}

\begin{figure}[htb]
        \centering
        \includegraphics[width=0.8\textwidth]{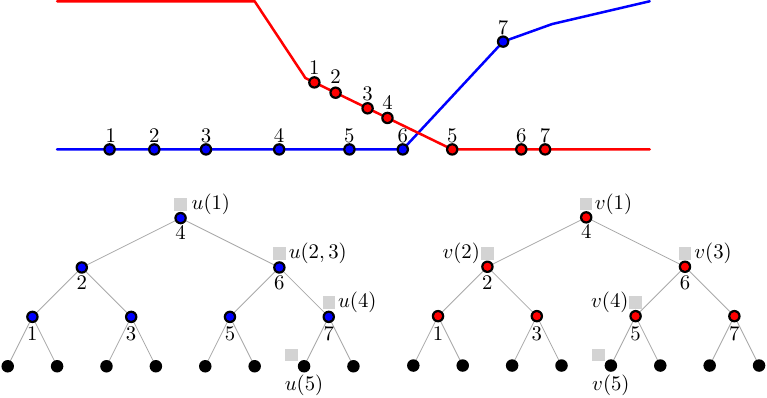}
        \caption{Top: the two lower envelopes $\cL_l$ and
	$\cL_r$ for 
	the pseudo-lines in Figure~\ref{f:intersection_point_demo2}. 
	Bottom: the corresponding trees $\Lambda_l$ and $\Lambda_r$.
	The labels $u(i)$ and $v(i)$ indicate the position 
	of the pointers $u$ and $v$ at step $i$, during the search.}
    	\label{f:intersection_point_demo} 
\end{figure}

\begin{table}[htb]
   \centering
    \begin{tabular}{|c|cc|cc|c|}
    \hline 
        Step  & $u$ & $v$ & uStack & vStack & Procedure case \\\hline
        1  & 4 & 4 & $\emptyset$ & $\emptyset$ & Case~3 \\
        2  & 6 & 2 & 4 & 4 & Case~2 $\rightarrow$ Case~2 \\
        3  & 6 & 6 & 4 & $\emptyset$ & Case~3 \\
        4  & 7 & 5 & 4, 6 & 6 & Case~1 $\rightarrow$ Case~3 \\
        5  & 7* & 5* & 4, 6 & 6, 5 & Case~3 $\rightarrow$ End\\\hline
    \end{tabular}
    \caption{The progress of the search for the example in
    Figure~\ref{f:intersection_point_demo}.}
    \label{fig:algodemo}
\end{table}

\bigskip
\end{document}